\documentclass{article}
\usepackage{iclr2027_conference,times}

\usepackage{hyperref}
\usepackage{url}

\title{Parameterized Complexity of $L_p$-Lipschitz Constants for Input Convex Neural Networks and $L_p$-Norm Maximization over Zonotopes}

\newcommand{\affilmark}[1]{\textsuperscript{\normalfont #1}}
\newlength{\authorblockwidth}
\author{
  \parbox[t]{\authorblockwidth}{
    \centering
    \bfseries
    
    \begin{tabular}{@{}
        c@{\hspace{1.15em}}
        c@{\hspace{1.15em}}
        c@{\hspace{1.15em}}
        c@{\hspace{1.15em}}
      c@{}}
      Aritra Das\affilmark{1}
      &
      Vincent Froese\affilmark{2}
      &
      Moritz Grillo\affilmark{3}
      &
      Debayan Gupta\affilmark{1}
      &
      Christoph Hertrich\affilmark{4}
      \\
      \multicolumn{5}{c}{
      Tharrshann Jayan Logarajah\affilmark{5}
      \hspace{1.3em}
      Georg Loho\affilmark{6}
      \hspace{1.3em}
      Mihir More\affilmark{1}
      \hspace{1.3em}
      Moritz Stargalla\affilmark{4}
      }
    \end{tabular}
    \\[7pt]
    \normalfont\footnotesize
    \affilmark{1}\,Ashoka University, Truth Audit Labs
    \qquad
    \affilmark{2}\,Technische Universität Berlin, Algorithmics and Computational Complexity Research Group
    \\[1.5pt]
    \affilmark{3}\,Max Planck Institute for Mathematics in the Sciences
    \qquad
    \affilmark{4}\,University of Technology Nuremberg
    \\[1.5pt]
    \affilmark{5}\,University College London
    \qquad
    \affilmark{6}\,Freie Universität Berlin
    \\[6pt]
    \fontsize{8}{9.5}\selectfont
    \texttt{\{aritra.das,debayan.gupta,mihir.more\}@ashoka.edu.in}
    \hspace{1em}
    \texttt{vincent.froese@tu-berlin.de}
    \\[0.5pt]
    \texttt{moritz.grillo@mis.mpg.de}
    \hspace{1em}
    \texttt{\{christoph.hertrich,moritz.stargalla\}@utn.de}
    \\[0.5pt]
    \texttt{tjlogarajah@outlook.com}
    \hspace{1em}
    \texttt{georg.loho@math.fu-berlin.de}
  }
}

\iclrfinalcopy

\usepackage{amsmath,amssymb,amsthm,mathtools}
\usepackage{thmtools}
\usepackage{cleveref}

\newtheorem{theorem}{Theorem}[section]
\newtheorem{corollary}[theorem]{Corollary}
\newtheorem{proposition}[theorem]{Proposition}
\newtheorem{lemma}[theorem]{Lemma}

\newcommand{\R}{\mathbb{R}}
\newcommand{\Q}{\mathbb{Q}}
\newcommand{\N}{\mathbb{N}}
\newcommand{\norm}[1]{\left\lVert #1\right\rVert}
\newcommand{\conv}{\operatorname{conv}}

\newcommand{\poly}{\operatorname{poly}}
\DeclareMathOperator{\cone}{cone}
\DeclareMathOperator{\dist}{dist}

\newcommand{\problemdef}[3]{
	\begin{center}
		\begin{minipage}{0.95\linewidth}
			\noindent
			\textsc{#1}
			\vspace{2pt}\\
			\setlength{\tabcolsep}{3pt}
			\begin{tabular}{@{}l p{0.85\linewidth}@{}}
				\textbf{Input:}    & #2 \\
				\textbf{Question:} & #3
			\end{tabular}
		\end{minipage}
	\end{center}
}
\begin{document}

\maketitle

\begin{abstract}
Lipschitz constants are a standard way to quantify the sensitivity of neural networks to small input perturbations, but computing them is difficult even for shallow ReLU networks.
Computing the $L_p$-Lipschitz constant for two-layer input-convex neural networks (ICNNs) is equivalent to maximizing the dual norm over a zonotope.
While $L_1$- and $L_\infty$-norm maximization on zonotopes are fixed-parameter tractable and polynomial-time solvable, respectively, the parameterized complexity of the remaining $L_p$-norms was open.

We prove that, for every fixed $p\in (1,\infty)\cap \Q$, maximizing the $L_p$-norm over a zonotope in $\R^d$ is W[1]-hard with respect to the dimension $d$.
Moreover, our hardness results imply that brute-force enumeration algorithms are essentially optimal for this problem under the Exponential Time Hypothesis.
By duality, the same hardness results hold for computing the $L_p$-Lipschitz constant of a two-layer ReLU ICNN.

Our paper resolves an open problem posted at COLT~'25 \citep{pmlr-v291-froese25b}.
There are several independent concurrent papers resolving the same problem.
Our paper prioritizes a clear exposition of the underlying geometric ideas and conceptual intuitions behind the proof. Additionally, we explicitly describe our research process including the use of LLMs.
\end{abstract}

\section{Introduction}

Neural networks with rectified linear unit (ReLU) activations play an important role in machine learning.
In practice, such networks are trained on finite datasets and are expected to generalize reliably to unseen inputs.
Nevertheless, even small input perturbations can lead to large changes in the output \citep{szegedy2014}.
The \emph{$L_p$-Lipschitz constant} of a neural network gives a quantitative measure of how sensitive the network is to small input perturbations.
Neural networks with small Lipschitz constants have been observed to be more robust to adversarial attacks and to have better generalization properties. The problem of approximating and computing the exact $L_p$-Lipschitz constant of neural networks has therefore received substantial attention \citep{VS18,Weng18,fazlyab2019efficient,JD20}.

Computing the $L_p$-Lipschitz constant exactly for $p\in (0,\infty]$ is NP-hard already for two-layer ReLU networks \citep{VS18,froese2026parameterizedhardnesszonotopecontainment}, and approximation within any multiplicative factor is NP-hard for three-layer ReLU networks \citep{JD20,FGS25}.
The hardness of this problem is closely connected to the curse of dimensionality which also afflicts many other problems in machine learning: the size of the search space of these problems grows exponentially with the input dimension~$d$.
In practice, however, high-dimensional data is often assumed to lie near some low-dimensional submanifold of the input space.
Therefore, a natural follow-up question is whether these problems become tractable when the input dimension is small.

This viewpoint has motivated recent work on the \emph{parameterized complexity} of neural network problems, including training \citep{AroraBMM18,FHN22,brand2023new,FH23,ganian2026quantized} and verification \citep{FGS25,froese2026parameterizedhardnesszonotopecontainment}; see also the survey \citep{ganian2026parameterized}.
When allowing arbitrary weights, computing the $L_p$-Lipschitz constant is W[1]-hard with respect to the input dimension $d$ already for two-layer ReLU networks, approximation is W[1]-hard for three-layer ReLU networks, and the problems are not solvable in time $n^{o(d)}\poly(N)$ (where $n$ is the width and $N$ the encoding size) under the Exponential Time Hypothesis (ETH), matching the running times of brute-force enumeration algorithms \citet{froese2026parameterizedhardnesszonotopecontainment}.
The known hardness reductions, however, use both positive and negative output weights and therefore do not imply hardness when putting additional restrictions on the network architecture.

One important such restricted architecture is an \emph{input-convex neural network} (ICNN) with ReLU activations, where the weights of all but the first layer are restricted to be nonnegative \citep{amos2017input}.
This restriction ensures that ICNNs compute convex functions.
ICNNs provide a direct way to incorporate prior knowledge into the network architecture and have been used, for example, in optimal control, optimal transport, and convex optimization \citep{chen2018optimal,makkuva2020optimal,huang2021convex}.
Restricting to ICNNs can make computational problems easier.
For example, minimizing the function represented by a ReLU ICNN is polynomial-time solvable, whereas the same problem is NP-hard and W[1]-hard with respect to the input dimension for general ReLU networks \citep{amos2017input,froese2026parameterizedhardnesszonotopecontainment}.
Further, the $L_1$-Lipschitz constant of ReLU ICNNs can be computed in polynomial time, and the $L_\infty$-Lipschitz constant is fixed-parameter tractable with respect to the input dimension, while the same problems for general ReLU networks are NP-hard and W[1]-hard with respect to the input dimension \citep{froese2026parameterizedhardnesszonotopecontainment}.

This raises the natural question whether the restriction to ICNNs also makes the problem tractable for the remaining $L_p$-norms with $p\in (1,\infty)$.
In particular, this range includes the $L_2$-norm, which, unlike the $L_1$- and $L_\infty$-norms, is invariant under orthogonal transformations and has been studied extensively in the context of robustness of neural networks \citep{moosavi2016deepfool,cisse2017parseval,cohen2019certified}.
Moreover, the question whether the $L_p$-Lipschitz constant of a (two-layer) ReLU ICNN is fixed-parameter tractable with respect to the input dimension $d$ for $p\in (1,\infty)$ has been posed as an open problem at COLT 2025 \citep{pmlr-v291-froese25b}, and, in the equivalent language of norm-maximization over zonotopes, in \citep{Shenmaier20}.

For two-layer ReLU ICNNs with input dimension $d$, computing the $L_p$-Lipschitz constant is equivalent to maximizing the $L_q$-norm over a zonotope in $\R^d$, where $p$ and $q$ are conjugate exponents \citep{pmlr-v291-froese25b}.
The $L_p$-norm maximization problem over zonotopes was studied long before this connection to ICNNs was known.
\citet{BGKL90} showed NP-hardness for every $p\in \N_{\geq 1}$.
For the $L_2$-norm, the problem is known as unconstrained binary quadratic maximization for positive-semidefinite matrices of rank $d$ \citep{FFL05}.
For $p\in [1,\infty]$, the problem has been studied as the \emph{Longest Vector Sum} problem, with applications in pattern recognition, clustering, signal processing, and political analysis \citep{Pyatkin10,Shenmaier18,Shenmaier20}.
Known results include NP-hardness and inapproximability \citep{Pyatkin10,Shenmaier20}, fixed-parameter tractability with respect to $d$ for $p=1$, polynomial time solvability for $p=\infty$ \citep{BP07}, and exact algorithms with running time $n^{d-1}\poly(N)$ \citep{FFL05,Shenmaier20}.
For $p\in (1,\infty)$, fixed-parameter $(1-\varepsilon)$-approximation algorithms are known for every fixed $\varepsilon>0$ \citep{Shenmaier18,froese2026parameterizedhardnesszonotopecontainment}.
A zonotope generated by $n$ vectors in $\mathbb{R}^d$ has $\mathcal{O}(n^{d-1})$ vertices, and the best-known exact algorithms for $p\in (1,\infty)$ essentially enumerate all vertices in $n^{d-1}\poly(N)$ time (where $N$ is the input bit-length) \citep{FFL05,froese2026parameterizedhardnesszonotopecontainment}.

\subsection{Our Contributions}
We resolve this question negatively by proving W[1]-hardness with respect to the parameter input dimension~$d$, which excludes fixed-parameter tractability under standard complexity assumptions.
Moreover, assuming the ETH, we show that solving these problems via brute-force enumeration algorithms that enumerate the vertices of zonotopes or the linear regions of neural networks is essentially optimal.

In \Cref{sec:L2hardness}, we give a reduction from \textsc{Multicolored Clique} to \textsc{$L_2$-Max on Zonotopes} (see \Cref{sec:prelims} for a formal definition) in which the dimension $d$ depends linearly on the requested clique size $k$.
The reduction proves W[1]-hardness and tight running time lower bounds based on the ETH.
The main difficulty in the reductions is to keep the dimension linear in $k$, in contrast to classical NP-hardness reductions where the dimension can grow without restriction.
In \Cref{sec:reduction-explicit}, we give another reduction that provides different insights compared to the one in \Cref{sec:L2hardness}: the one in \Cref{sec:L2hardness} is non-explicit (the generators are obtained via linear programs) but contains interesting geometric insights, whereas the one in \Cref{sec:reduction-explicit} is explicit and uses elementary techniques.

In \Cref{sec:Lphardness}, we then transfer the construction for the $L_2$-norm to general $L_p$-norms and prove W[1]-hardness and tight ETH-based running time lower bounds for \textsc{$L_p^p$-Max on Zonotopes} and \textsc{$L_p$-Max on Zonotopes} for every fixed $p\in (1,\infty)\cap \Q$.
We describe the equivalence between computing the $L_p$-Lipschitz constant of two-layer ICNNs and $L_q$-norm maximization over zonotopes in \Cref{sec:icnn_lipschitz}.
This equivalence yields W[1]-hardness and tight ETH-based running time lower bounds for \textsc{Two-Layer ICNN $L_p$-Lipschitz Constant}.

\subsection{Concurrent Results}
Independently of our work, a proof for W[1]-hardness of \textsc{$L_2$-Max on Zonotopes} was recently obtained by an agentic AI pipeline.\footnote{\label{fn:ai-pipeline} The proof was made available at \url{https://github.com/Pengbinghui/pipeline-math/}.}
The proof gives a reduction from \textsc{Multicolored Clique} with dimension $2k+1$, and works with support functions of centrally symmetric zonotopes.
We adopt a similar viewpoint in \Cref{sec:L2hardness}, but present a different reduction with $2k$ dimensions using theoretical insights from polyhedral geometry.
We discuss our own usage of and experiences with AI in \Cref{sec:AIprocess}.
Very recently, again independently of our work, proofs showing W[1]-hardness of \textsc{$L_2$-Max on Zonotopes} \citep{dewasurendra2026convex} and W[1]-hardness of \textsc{$L_p^p$-Max on Zonotopes} for every fixed rational $p\in (1,\infty)$ \citep{cao2026ellpnorm} were published.
\citet{dewasurendra2026convex} give a reduction from \textsc{Partitioned Subgraph Isomorphism} with dimension $11k+1$, which shows that under the ETH, no $\rho(k)n^{o(k/\log k)}$ time algorithm exists for \textsc{$L_2$-Max on Zonotopes}.
\citet{cao2026ellpnorm} give a reduction from binary CSP and then specialize it to a reduction from \textsc{Multicolored Clique} with dimension $2k+1$, which shows that under the ETH, no $\rho(k)n^{o(k)}$ time algorithm exists for \textsc{$L_p^p$-Max on Zonotopes}.
We obtain the same running time lower bound.
Moreover, \citet{cao2026ellpnorm} also present deterministic fixed-parameter $(1-\varepsilon)$-approximation algorithms for every fixed $\varepsilon>0$ running in time $f_p(d)\varepsilon^{-(d-1)/2}\poly(N, \log(1/\varepsilon))$, and rule out algorithms with $(1/\varepsilon)^{o(d)}$ dependence under the ETH.
It is worth noting that all proofs, including ours, construct points lying exactly or approximately on the two-dimensional $L_2$-unit circle (or $L_p$-unit circle in \citet{cao2026ellpnorm}) that define a zonotope and then use this zonotope to encode choices of nodes or variables, see Section~2.1 in the AI agentic pipeline proof (\Cref{fn:ai-pipeline}), Lemma~4 of \citet{dewasurendra2026convex}, and Lemma~3.2 of \citet{cao2026ellpnorm}.
Despite these similarities, the details of the reductions are all quite different, as outlined above.
In particular, we present a reduction from \textsc{Multicolored Clique} with dimension $2k$ instead of $11k+1$ or $2k+1$ for $p=2$ in \Cref{sec:L2hardness}, and provide geometric intuition for all of our proofs.

\subsection{Further Related Work}
Computing Lipschitz constants of two-layer ReLU ICNNs is not the only neural network problem that can be formulated in terms of maximizing certain norms over zonotopes:
for example, maximizing a zonotope-induced norm over a zonotope corresponds to the zonotope containment problem, which in turn is equivalent to the deciding whether a two-layer ReLU network attains a positive output value \citep{KULMBURG202184,froese2026parameterizedhardnesszonotopecontainment}.
This problem is also known to be W[1]-hard with respect to the input dimension $d$.

$L_p$-norm maximization on general polytopes in halfspace representation is W[1]-hard with respect to the dimension $d$ for $p\in \N_{\geq 2}$ and fixed-parameter tractable for $p=1$ \citep{KKW15}.
This result does not extend to zonotopes, since zonotopes are given in generator representation, and their halfspace representation can be exponentially larger than their generator representation.
The parameter dimension $d$ has also been considered for many other geometry problems.
Several reductions in this area, as ours, use specially constructed points on the two-dimensional $L_2$-unit sphere that are sometimes also called ``scaffolding point sets'', and embed them into higher dimensions \citep{GKR09,CGKMR11,KKW15,FHN22}.

\section{Preliminaries}
\label{sec:prelims}
\textbf{Notation.}
We write $\N=\{0,1,\dots\}$ and set $[n]\coloneq \{1,\dots, n\}$ for $n\in \N_{\geq 1}$.
A function $f\colon\R^d\to \R$~is \emph{positively homogeneous} if $f(\lambda x)=\lambda f(x)$ for all $x\in \R^d$ and $\lambda\geq 0$.
A vector $x\in \R^d$ is \emph{decreasing} if $x_1\geq \dots \geq x_d$.
For a set $T\subseteq \R^d$, define $\operatorname{dist}(y,T)\coloneqq\min_{t\in T}\|y-t\|_2$.

\textbf{$L_p$-Norms.}
For $p\in [1,\infty)$ and a vector $x\in \R^d$, the \emph{$L_p$-norm} of $x$ is $\norm{x}_p\coloneq \left(\sum_{i=1}^d|x_i|^p\right)^{1/p}$.
For $p=\infty$, the \emph{$L_\infty$-norm} of $x$ is $\norm{x}_\infty\coloneq \max_{i\in [d]}|x_i|$.
For $p\in[1,\infty]$, the \emph{$L_p$-Lipschitz constant} of a function $f:\R^d\to \R^m$ is $L_p(f)\coloneq \sup_{x\neq y}\frac{\norm{f(x)-f(y)}_p}{\norm{x-y}_p}$.

\textbf{Polyhedra and Cones.}
A \emph{polyhedron} $P$ is the intersection of finitely many closed halfspaces $P=\{x\in \R^d: Ax\leq b\}$.
A \emph{face} of $P$ is either the empty set or the set of maximizers $\arg\max\{c^\top x: x\in P\}$ of a linear function over $P$.
Faces of dimension zero are called \emph{vertices}; faces of dimension $\dim(P)-1$ are called \emph{facets}.
A \emph{polytope} is a bounded polyhedron; equivalently, it is the convex hull of finitely many points.
A \emph{cone} $C\subseteq \R^d$ is a set where $x\in C$ implies $\lambda x\in C$ for all $\lambda \geq 0$.
The \emph{conical hull} of a set $S\subseteq \R^d$ is the cone $\cone(S) \coloneq \{\sum_{i=1}^k \lambda_i x_i: x_i\in S, \lambda_i\in \R_{\geq 0},k\in \N\}$.
A cone is \emph{pointed} if it does not contain a line, and \emph{finitely generated} if it is the conical hull of a finite set.
An \emph{extreme ray} $R$ of a cone $C$ is a subset $R\subseteq C$ of the form $R=\R_{\geq 0} v$, with $v\neq 0$ such that all $x,y\in C$ with $x+y\in R$ satisfy $x,y\in R$.
We sometimes say that $v$ is an extreme ray if $\R_{\geq 0} v$ is an extreme ray.
Every pointed and finitely generated cone $C$ is the conical hull of its extreme rays.

\textbf{Zonotopes.}
Given a matrix $G=(g_1,\ldots,g_n)\in\R^{d\times n}$, the associated \emph{zonotope} is the polytope
$
Z(G)\coloneq\sum_{i=1}^n \conv(\{0,g_i\})=\left\{\sum_{i=1}^n\lambda_ig_i: \lambda\in [0,1]^n\right\}.
$
The matrix~$G$ is the \emph{generator matrix} of $Z(G)$, and its columns are the \emph{generators} of $Z(G)$.
We also define the \emph{symmetric zonotope} $Z_{\pm}(G)\coloneq \sum_{i=1}^n\conv\left(\{-g_i,g_i\}\right)=\left\{\sum_{i=1}^n\lambda g_i: \lambda\in [-1,1]^n\right\}$.
We have $Z_{\pm}(G)=Z([G,-G])$.

\textbf{ReLU Neural Networks.}
A two-layer ReLU network with scalar output is given by input weights $W \in \R^{n \times d}$ with rows $w_1^\top,\dots, w_n^\top\in \R^{1\times d}$, a bias vector $b \in \R^n$, output weights $a \in \R^n$, and an output bias $B\in \R$.
It computes the continuous piecewise-linear (CPWL) function $f \colon \R^d \to \R$ given by $f(x) = B+\sum_{i=1}^n a_i\max\{0,w_i^\top x + b_i\}$, where $\max\{0,x\}$ is the ReLU activation function.
After deleting zero output weights and rescaling the weights, that is, replacing $(w_i,b_i)$ with $|a_i|(w_i,b_i)$ and $a_i$ with $\text{sign}(a_i)\in \{-1,1\}$, we can assume $a\in \{-1,1\}^n$ without loss of generality.
If the network is bias-free, that is, $b=0$ and $B=0$, then $f$ is positively homogeneous.
If $a\geq 0$, then $f$ is convex and the network is called a two-layer ReLU \emph{input-convex neural network} (ICNN).

\textbf{Parameterized Complexity.}
A \emph{parameterized} problem consists of instances~$(x,k)$, where~$x$ is the classical problem instance and~$k\in\N$ is a \emph{parameter}.
A parameterized problem is \emph{fixed-parameter tractable} (contained in the class FPT) if it is solvable in~$f(k)\cdot |x|^{\mathcal{O}(1)}$ time for an arbitrary function~$f$ solely depending on~$k$.
The class W[1] contains parameterized problems which are presumably not in FPT. Such \emph{W[1]-hard} problems are defined via \emph{parameterized reductions}: A parameterized reduction from $L$ to $L'$ is an algorithm that maps an instance $(x,k)$ in $f(k)\cdot |x|^{\mathcal{O}(1)}$ time to an instance $(x',k')$ such that $k'\leq g(k)$ for some function $g$ and $(x,k)\in L$ holds if and only if $(x',k')\in L'$.
We refer to \citet{DF13} for further details on parameterized complexity.
The \emph{Exponential Time Hypothesis} (ETH) \citep{IP01} states that 3-SAT on $n$ variables cannot be solved in $2^{o(n)}$ time.
The ETH implies FPT $\neq$ W[1] as well as running time lower bounds for various problems, such as that \textsc{Clique} cannot be solved in $\rho(k)n^{o(k)}$ time, where $k$ is the size of the requested clique and $n$ is the number of nodes in the graph \citep{Cygan15}.

\subsection{Problem Definitions}
For $p\in [1,\infty]$, we consider the following two problems.
\problemdef
{Two-Layer ICNN $L_p$-Lipschitz Constant}
{Weights $W=(w_1,\dots, w_n)^\top\in \Q^{n\times d}$, biases $b\in \Q^n,\, B\in \Q$, and $L\in \Q$.}
{Is $L_p(f)\geq L$ with $f(x)=B+\sum_{i=1}^n \max\{0,w_i^\top x + b_i\}$?}

\problemdef
{$L_p$-Max on Zonotopes}
{A rational generator matrix $G\in\Q^{d\times n}$ and $L\in\Q$.}
{Is $\max_{z\in Z(G)} \norm{z}_p \geq L$?}

For $p\in [1,\infty)$, \textsc{$L_p^p$-Max on Zonotopes} denotes the same problem with the question $\max_{z\in Z(G)} \norm{z}_p^p \geq L$, for which we use the following alternative formulation, which follows from the fact that any zonotope is a projection of a hypercube $[0,1]^n$.
All missing proofs, as well as detailed proofs for all proof sketches provided in the main text, can be found in \Cref{app:proofs}.
\begin{lemma}\label{lem:l2max_binaryformulation}
	For every generator matrix $G=(g_1,\ldots,g_n)\in \R^{d\times n}$ and every $p\in [1,\infty)$,
	\[
	\max_{z\in Z(G)} \norm{z}_p^p
	=
	\max_{\sigma\in\{0,1\}^n}\norm{\sum_{i=1}^n \sigma_i g_i}_p^p.
	\]
\end{lemma}

\section{\texorpdfstring{$L_p$}{L_p}-Lipschitz-constant for Input-convex  ReLU Networks}
\label{sec:icnn_lipschitz}
In this section, we discuss the equivalence between \textsc{$L_p$-Max on Zonotopes} and \textsc{Two-Layer ICNN $L_q$-Lipschitz Constant}, where $p,q\in[1,\infty]$ are conjugate exponents, that is, $1/p+1/q=1$.
This equivalence allows us to later transfer the hardness for \textsc{$L_p$-Max on Zonotopes} to \textsc{Two-Layer ICNN $L_q$-Lipschitz Constant}.
The equivalence is a consequence of a well-known duality between convex and positively homogeneous CPWL functions and polytopes from tropical geometry, which we now briefly sketch, see, e.g., \citep{zhang2018tropical} for more details.

Every ReLU ICNN computes a convex CPWL function of the form $f\colon\R^d\to \R,\; x\mapsto \max\{a_1^\top x +b_1,\dots, a_k^\top x+b_k\}$, which subdivides $\R^d$ into a set of polyhedral regions, also called \emph{linear regions}, on which the function is affine linear.
More precisely, there exists a finite set $\mathcal{C}$ of disjoint polyhedra with $\R^d=\bigcup_{C\in \mathcal{C}}C$ and $f(x)=a_C^\top x+b_C$ for all $x\in C$.
Then, using the well-known equality $L_p(g)=\max_{x\in \R^d}\|\nabla g(x)\|_q$ for smooth functions $g\colon\R^d\to \R$ \citep{JD20}, we obtain that the $L_p$-Lipschitz constant of $f$ on the region $C$ is equal to $\|a_C\|_q$, which implies
\[
L_p(f)=\max_{C\in \mathcal{C}} \norm{a_C}_q.
\]
Hence, computing the $L_p$-Lipschitz constant amounts to maximizing the dual norm of the coefficient vectors of the linear regions of $f$.
Observe that the function $f$ of the ICNN has the same $L_p$-Lipschitz constant as the positively homogeneous function $\max\{a_1^\top x,\dots, a_k^\top x\}$ computed by the same ICNN with all biases set to zero.
This allows us to restrict to bias-free ICNNs without loss of generality.

There is a bijection $\varphi$ between the set $\mathcal{F}_d$ of convex and positively homogeneous CPWL functions from $\R^d$ to $\R$ computed by bias-free ReLU ICNNs and the set $\mathcal{P}_d$ of polytopes in $\R^d$, given by $\varphi\left(\max\{a_1^\top x,\dots, a_k^\top x\}\right)=\conv\{a_i: i\in [k]\}$.
The inverse $\varphi^{-1}\colon\mathcal{P}_d\to \mathcal{F}_d$ gives the \emph{support function} $\varphi^{-1}(P)(x)=\max_{y\in P}y^\top x$ of a polytope $P\in \mathcal{P}_d$.
The bijection satisfies $\varphi(f+g)=\varphi(f)+\varphi(g)$ (where the $+$ becomes a Minkowski sum) and $\varphi(\max\{f,g\})=\conv\{\varphi(f)\cup \varphi(g)\}$ for all $f,g\in \mathcal{F}_d$.

\begin{figure}
    \centering
    \resizebox{0.78\linewidth}{!}{
	\includegraphics[page=1, width=\textwidth]{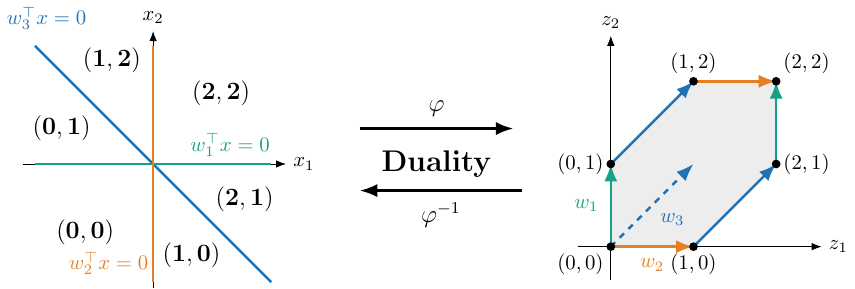}
    }
    \caption{
    Illustration of the duality between zonotopes and bias-free two-layer ICNNs for $W=(w_1,w_2,w_3)^\top$ with $w_1=(0,1)^\top$, $w_2=(1,0)^\top$, and $w_3=(1,1)^\top$.
    The ICNN computes $f\colon\R^2\to \R,\, x\mapsto \sum_{i=1}^3\max\{0,w_i^\top x\}$ (left) and the corresponding zonotope is $Z(W^\top)$ (right).
    The gradient $\nabla f$ on a linear region (bold vector) corresponds to exactly one of the six vertices of $Z(W^\top)$.
    }
    \label{fig:icnn_zonotope_duality}
\end{figure}

Applying this bijection to a bias-free two-layer ReLU ICNN that computes $f\colon\R^d\to \R,\; x\mapsto \sum_{i=1}^n\max\{0,w_i^\top x\}$ with weight matrix $W=(w_1,\dots, w_n)^\top\in \R^{n\times d}$ shows that the polytope corresponding to $f$ is
\[
\varphi(f)
=\varphi\left(\sum_{i=1}^n\max\{0,w_i^\top x\}\right)
=\sum_{i=1}^n\varphi\left(\max\{0,w_i^\top x\}\right)
=\sum_{i=1}^n\conv(\{0,w_i\})
=Z(W^\top),
\]
a zonotope.
It is a well-known fact that the vertices of $Z(W^\top)$ are in bijection to the regions of the hyperplane arrangement $\mathcal{H}_W=\left(\{x\in \R^d: w_i^\top x= 0\}\right)_{i\in [n]}$ defined by $W$ \citep{Ziegler95}, which are precisely the linear regions of $f$.
Moreover, the coefficient vector $a_C$ of a linear region $C\in \mathcal{C}$ is precisely the vertex $v$ of $Z(W^\top)$ that corresponds to the linear region (see \Cref{fig:icnn_zonotope_duality}).
Since every vertex $v$ of $Z(W^\top)$ is of the form $\sum_{i=1}^n \sigma_i w_i$ for some $\sigma_i\in \{0,1\}$ \citep{Ziegler95}, we obtain
\[
L_p(f)
=\max_{\sigma\in \{0,1\}^n}\norm{\sum_{i=1}^n \sigma_i w_i}_q
=\max_{z\in Z(W^\top)}\|z\|_q.
\]

\section{Hardness Reduction for \texorpdfstring{\textsc{$L_2$-Max on Zonotopes}}{L_2-Max on Zonotopes}}
\label{sec:L2hardness}

We now give a reduction for \textsc{$L_2$-Max on Zonotopes}.
We reduce from the following problem.
\problemdef{Multicolored Clique}
{A graph $G=(V=V_1\dot\cup\cdots\dot\cup V_k, E)$, where each node in $V_i$ has color $i$.}
{Does $G$ have a $k$-colored clique (a clique with exactly one node of each color)?}
\textsc{Multicolored Clique} is NP-hard, W[1]-hard with respect to $k$, and not solvable in time $\rho(k) |V|^{o(k)}$ for any computable function $\rho$ assuming the ETH \citep{Cygan15}.

The reduction is based on the support function $h_{Z_{\pm}(A)}\colon\R^{2k}\to \R,\; x\mapsto \sum_{i=1}^n|a_i^\top x|$ of a centrally symmetric zonotope $Z_{\pm}(A)\subseteq \R^{2k}$, which satisfies, analogously to \Cref{sec:icnn_lipschitz},
\[
\max_{z\in Z_{\pm}(A)}\|z\|_2
=
\max_{z\in Z_{\pm}(A)}\max_{\|y\|_2=1} |y^\top z|
=
\max_{\|y\|_2=1}\max_{z\in Z_{\pm}(A)}|y^\top z|
=
\max_{\|y\|_2=1}h_{Z_{\pm}(A)}(y).
\]

\textbf{Interpolation.}
We first prove a structural result about support functions of centrally symmetric zonotopes.
For a set $S=\{s_1,\ldots,s_M\}\subset \mathbb Q^r\setminus \{0\}$ and $w\in\R^r$, define $\Phi_S(w)\coloneq (|w^\top s_1|,\dots, |w^\top s_M|)\in \R^M$ as the vector of values of the support function of the zonotope $\conv(\{-w,w\})$ evaluated at all points in~$S$ and define the cone
\[
C(S)\coloneq\cone\{\Phi_S(w): w\in \R^r\}\subseteq \R^M,
\]
which corresponds to all values that a support function of a centrally symmetric zonotope can attain on the vectors in $S$.
To see this, note that $b\in C(S)$ if and only if there are vectors $w_1,\dots, w_q\in \R^r$ and nonnegative scalars $\lambda_1,\dots, \lambda_q \in \R_{\geq 0}$ such that $b=\sum_{t=1}^q \lambda_t \Phi_S(w_t)$, which is equivalent to $b_i=\sum_{t=1}^q  \lambda_t| w_t^\top s_i|=h_{Z_\pm ([\lambda_1 w_1,\dots, \lambda_q w_q])}(s_i)$ for all $i\in [M]$.
\begin{lemma}\label{lem:cone_structure}
    Fix $r\ge 1$.
    Let $S=\{s_1,\ldots,s_M\}\subset \mathbb Q^r\setminus \{0\}$ be a set of pairwise linearly independent vectors that all have equal Euclidean norm.
    Then, the cone $C(S)$ is finitely generated and contains the all-one vector in its interior, that is, $\mathbf{1}\in \operatorname{int} C(S)$.
    Moreover, we can compute the set of extreme rays of $C(S)$ in polynomial time.
\end{lemma}
\begin{proof}[Proof sketch]
    Let $U=\operatorname{span}(S)$ and consider the hyperplane arrangement
    $\mathcal{H}\coloneqq\left(H_m\right)_{m\in [M]}$,
    where
    $H_m\coloneq\{w\in U:s_m^\top w=0\}$.
    Since we are in fixed dimension $\le r$, we can compute all faces of dimension one $\R w_1, \dots, \R w_q$ of~$\mathcal{H}$ in polynomial time, as each such face is the intersection of $\dim(U)-1$ many linearly independent hyperplanes (of which there are at most $M^{\mathcal O(r)}$ many).
    We then show that $\mathcal{R}\coloneqq\{\Phi_S(w_t):t\in [q]\}$ is the set of extreme rays of $C(S)$.
    Since $C(S)\subseteq \R^M_{\geq 0}$, the cone is pointed and thus generated by its extreme rays, so $C(S)=\cone(\mathcal{R})$.
    We then prove that $\mathbf 1 \in \operatorname{int} C(S)$ (since all vectors in~$S$ have equal length) using a variant of the Farkas' lemma.
\end{proof}

See also \Cref{fig:cone} for a visual interpretation of the proof.
By \Cref{lem:cone_structure}, $\mathbf{1}\in \operatorname{int} C(S)$, which means that we can move a small amount from $\mathbf 1$ in any direction while staying in $C(S)$.
The following statement shows that for any direction $b$, we can compute generators of a zonotope whose support function attains value $1+\delta b_m$ on $s_m$ for a small $\delta>0$ in polynomial time.

\begin{figure}
    \centering
    \resizebox{0.9\linewidth}{!}{
	\includegraphics[page=3, width=\textwidth]{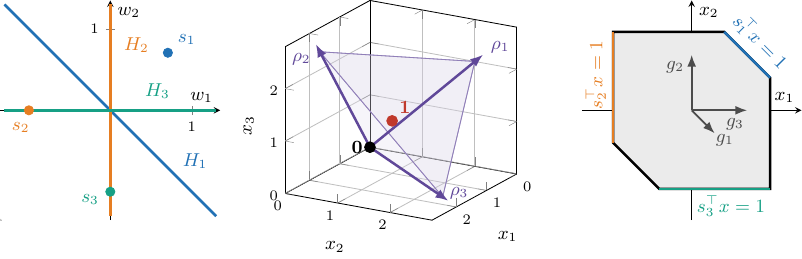}
    }
    \caption{
    Example for $S=\{(1,1)/\sqrt{2},(-1,0),(0,-1)\}$ (left).
    The faces of dimension one of $\mathcal H$ are $H_i=\R w_i$ for $i\in [3]$ with $w_1=(1,-1)$, $w_2=(0,1)$, $w_3=(1,0)$, and $C(S)=\cone\{\rho_1,\rho_2,\rho_3\}=\{(0,1,1),(1/\sqrt{2},0,1),(1/\sqrt{2},1,0)\}$ with $\rho_i=\Phi_S(w_i)$ for $i\in [3]$ (center).
    We have $\mathbf 1\in \operatorname{int}C(S)$ with $\mathbf 1 = (1-1/\sqrt{2})\rho_1+1/\sqrt{2}(\rho_2+\rho_3)$, and the corresponding zonotope is generated by $g_1=(1-1/\sqrt{2})(1,-1)$, $g_2=(0,1/\sqrt{2})$, and $g_3=(1/\sqrt{2},0)$ (right).
    }
    \label{fig:cone}
\end{figure}

\begin{proposition}\label{lem:interpolation}
    Fix $r\ge 1$.
    Let $S=\{s_1,\ldots,s_M\}\subset \mathbb Q^r\setminus \{0\}$ be  a set of pairwise linearly independent vectors that all have equal Euclidean norm, and let $b \in \mathbb{Q}^M$.
    Then one can compute in polynomial time rational vectors $g_1, \ldots, g_M \in \mathbb{Q}^r$ and a rational number $\delta > 0$ such that
    \begin{equation*}
        \sum_{t=1}^M | g_t^\top  s_m |
        =
        1 + \delta b_m
        \quad
        \text{ for every }m\in [M].
    \end{equation*}
    
\end{proposition}
\begin{proof}
    By \Cref{lem:cone_structure}, we can enumerate the extreme rays $\Phi_S(w_1),\dots, \Phi_S(w_q)$ of $C(S)$ in polynomial time.
    Since $\mathbf{1}\in \operatorname{int} C(S)$, the linear program
    \[
    \max\{\delta: \sum_{t=1}^q\lambda_t\Phi_S(w_t)=\mathbf{1}+\delta b,\, \lambda_t\geq 0,\,\delta\in [0,1]\}
    \]
    has an optimal solution $(\lambda,\delta)$ of polynomial encoding size with $\delta>0$.
    Since the linear program is of polynomial size, it can be solved in polynomial time and its optimal solution has polynomial encoding size \citep{Schrijver86}.
    By Caratheodory's theorem for cones in $\R^M$, there exists a solution where~$\lambda$ has at most $M$ nonzero entries, and we can find such a solution in polynomial time by iteratively eliminating linearly dependent extreme rays.
    We can then read off $\delta>0$ and $g_1,\dots, g_M\in \Q^r$ as $g_t=\lambda_t w_t$ from the optimal solution.
\end{proof}

In the reduction, we apply \Cref{lem:interpolation} in $r=4$ dimensions to a set $S$ containing $M=2n_in_j$ points encoding all node pairs of color~$i$ and~$j$, where the corresponding entry of~$b$ is 1 if the nodes are adjacent, and $0$ otherwise.
Note that this general technique shows how to encode the structure of an arbitrary bipartite graph into the support function of a centrally symmetric zonotope.

\begin{theorem}\label{prop:second_reduction_result}
    Given an instance $(G,k)$ of \textsc{Multicolored Clique}, one can construct in polynomial time a generator matrix $A_G\in \Q^{2k\times q}$ and rational numbers $L>\Delta>0$ such that
    \[
    \max_{z\in Z(A_G)}\|z\|_2
    \begin{cases}
    \geq L & \text{if $(G,k)$ is a yes-instance,}\\
    \leq L-\Delta & \text{otherwise.}
    \end{cases}
    \]
    Therefore, \textsc{$L_2$-Max on Zonotopes} is W[1]-hard with respect to $d$ and not solvable in $\rho(d) N^{o(d)}$ time (where $N$ is the input bit-length) for any computable function $\rho$ assuming the ETH.
\end{theorem}
\begin{proof}[Proof sketch]
We reduce from \textsc{Multicolored Clique}.
Given an instance $(G=(V=V_1\dot\cup\cdots\dot\cup V_k, E),k)$ of \textsc{Multicolored Clique}, we construct the support function of a $2k$-dimensional centrally symmetric zonotope.
Let $Y=Y_1\times \dots \times Y_k$, with $Y_i\cong \R^2$ for $i\in [k]$, so $Y\cong \R^{2k}$.
Let $\rho_i\colon\R^2\to Y$ denote the embedding into $Y_i$ with $\rho_i(x)=(\mathbf{0}_{2i-2},x,\mathbf{0}_{2k-2i})$.
Similarly, for $1\leq i<j\leq k$, let $\rho_{ij}\colon\R^2\times \R^2\to Y$ denote the embedding into $Y_i$ and $Y_j$ with $\rho_{ij}(x,y)=(\mathbf{0}_{2i-2},x,\mathbf{0}_{2j-2i-2},y,\mathbf{0}_{2k-2j})$.
We write $V_i=\{v_{i,1},\dots,v_{i,n_i}\}$.

We first construct a node term $F_V\colon Y \to \R$ whose maximum over the $L_2$-unit sphere is attained exactly at a finite set $\mathcal{T}\subset Y$, where every point of $\mathcal{T}$ encodes the choice of one node per color class.
We choose positive rational numbers $\omega_1, \dots, \omega_k\in \Q$ with $\sum_{i} \omega_i^2 = 1$ of polynomial bit-length.
For every node $v_{i,a}\in V_i$ with $a\in [n_i]$, we define the rational unit vector $u_{i, a} \coloneq \frac{(1-a^2 , 2a)}{1+a^2}\in \Q^2$.
Let $P_i\coloneq\conv\{\pm u_{i,a}: a\in [n_i]\}\subset \R^2$.
List the vertices of $P_i$ in cyclic order as $u_{i,1},\dots, u_{i,2n_i}$ with $u_{i,r+n_i}\coloneqq -u_{i,r}$ and define
$
g_{i,r}\coloneq (u_{i,r+1}-u_{i,r})/2$ for $r\in [n_i]$.
Then we have
\[
P_i=\sum_{r=1}^{n_i}\conv\left(\{-g_{i,r},g_{i,r}\}\right),
\qquad
h_{P_i}(y)=\sum_{r=1}^{n_i}\left|g_{i,r}^\top y\right|.
\]
We now embed each zonotope $P_i$ into $Y_i$ and define the zonotope
\[
Z\coloneqq\sum_{i=1}^k\rho_i(\omega_i P_i),
\qquad
F_V(y)\coloneq h_Z(y)=\sum_{i=1}^k \omega_i h_{P_i}(y_i).
\]
Finally, let $\mathcal{T}\coloneq \{(\sigma_1 \omega_1 u_{1,a_1},\dots, \sigma_k \omega_k u_{k,a_k}): a_i\in [n_i], \sigma\in \{\pm 1\}^k\}$.
Every $t\in \mathcal{T}$ is a unit vector and corresponds to the choice of nodes $\{u_{1,a_1},\dots,u_{k,a_k}\}$.
Further, for any unit vector $y\in Y$, $F_V(y)=1-\frac 12\operatorname{dist}(y,\mathcal{T})^2$ (see \Cref{lem:node_selection}), so moving away from $\mathcal{T}$ causes a quadratic loss.

Next, for every pair of colors, we use \Cref{lem:interpolation} to construct a zonotope support function that gives an additional gain exactly when the two selected nodes are adjacent.
Fix $1\leq i<j\leq k$.
For $a_i\in [n_i]$, $a_j\in [n_j]$, and $\tau\in \{\pm 1\}$, set $s_{a_i, a_j, \tau} \coloneq (\omega_iu_{i, a_i}, \tau \omega_j u_{j, a_j}) \in Y_i \times Y_j\cong \R^4$.
These $2n_in_j$ vectors are pairwise linearly independent and all have the same Euclidean norm.
Define the vector $b\in\{0,1\}^{2n_in_j}$ with corresponding entries $b_{a_i,a_j,\tau}=1$ if $\{v_{i,a_i},v_{j,a_j}\}\in E(G)$ and $b_{a_i,a_j,\tau}=0$ otherwise.
Applying \Cref{lem:interpolation} for $r=4$ to $S=\{s_{a_i, a_j, \tau}: a_i\in [n_i], a_j\in [n_j],\tau\in\{\pm1\}\}$ and $b$ yields rational vectors $g_{ij,1},\dots, g_{ij,q_{ij}}\in Y_i\times Y_j$ with $q_{ij}= 2n_in_j$, and a rational number $\delta_{ij}>0$.
Define $H_{ij}\colon Y_i\times Y_j\to \R$ as $
H_{ij}(x_i,x_j)
\coloneq
\sum_{t=1}^{q_{ij}}\left|g_{ij,t}^\top (x_i,x_j)\right|$.
Then
\[
H_{ij}(s_{a_i,a_j,\tau})=1+\delta_{ij}b_{a_i,a_j,\tau}
\qquad
\text{for }a_i\in [n_i], a_j\in [n_j],\tau\in\{\pm1\}.
\]
We now embed every vector $g_{ij,t}$ into $Y$ using $\rho_{ij}$ and define $F_E\colon Y \to \R$ as
\[
F_E(y)\coloneq \sum_{1\leq i<j\leq k}H_{ij}(y_i,y_j).
\]
Again, $F_E$ is the support function of a centrally symmetric zonotope.
Set $\delta_* \coloneqq \min_{1\le i < j \le k} \delta_{ij} > 0$ and $E_{\rm cl} \coloneqq \sum_{1 \le i < j \le k} (1 + \delta_{ij})$.
Now, consider a $t=(\sigma_1\omega_1 u_{1,a_1},\dots, \sigma_k\omega_k u_{k,a_k})\in \mathcal{T}$.
For a pair $i < j$, set $\tau_{ij}\coloneqq\sigma_i\sigma_j$.
Since $H_{ij}$ is even,
\[
H_{ij}(\sigma_i\omega_iu_{i,a_i},\sigma_j\omega_ju_{j,a_j})
=H_{ij}(s_{a_i,a_j,\tau_{ij}})
=1+\delta_{ij}b_{a_i,a_j,\tau_{ij}}
\]
and thus
\[
F_E(t)
\begin{cases}
    =E_{\rm cl} & \text{if $\{v_{1,a_1},\dots,v_{k,a_k}\}$ form a clique,}\\
    \leq E_{\rm cl}-\delta_* & \text{otherwise.}
\end{cases}
\]
Finally, we set $F = F_V +\eta F_E$ for a sufficiently small $\eta>0$.
If $G$ has a $k$-colored clique $\{v_{1,a_1},\dots, v_{k,a_k}\}$, then $t\coloneqq(\omega_1 u_{1,a_1},\dots, \omega_k u_{k,a_k})\in \mathcal{T}$ gives $F(t)=1+\eta E_{\rm cl}$.
Conversely, if $G$ has no $k$-colored clique, we then show that for any unit vector $y\in Y$ with $D\coloneq \dist(y,\mathcal T)$, we have $F(y)\leq 1+\eta E_{\rm cl}-\eta \delta_*+\eta \Lambda D-D^2/2$.
By choosing $\eta$ sufficiently small, we can then prove that $F(y)\leq 1+\eta E_{\rm cl}-\eta \delta_*/2$, which allows us to distinguish yes- and no-instances.
Since $F$ is the support function of a centrally symmetric zonotope, maximizing~$F$ over the $L_2$-unit sphere is equivalent to maximizing the $L_2$-norm over this zonotope.

The reduction clearly runs in polynomial time and the dimension is $d=2k$, proving W[1]-hardness.
Consequently, an algorithm with running time $\rho(d) N^{o(d)}$ would solve \textsc{Multicolored Clique} in time $\rho(2k)n^{o(k)}$, contradicting the ETH.
\end{proof}

\section{Hardness for \texorpdfstring{$L_p$}{Lp}-norms with \texorpdfstring{$p\in (1,\infty)$}{p ∈ (1,∞)}}
\label{sec:Lphardness}
\Cref{prop:second_reduction_result} proves hardness for the $L_2$-case and therefore also for the $L_2^2$-case (simply square the thresholds).
The following proposition shows that we can generalize the $L_2^2$-case to the $L_p^p$-case.

\begin{proposition}\label{prop:L2toLpTransfer}
Fix $p\in(1,\infty)\cap\Q$. Given a matrix
$A\in\Q^{d\times q}$ and rational numbers $L>\Delta>0$, one can
construct in polynomial time a matrix $B\in\Q^{2d\times(q+1)}$ and $L^*,\Delta^*\in \Q_{>0}$ such
that
\begin{align}
\label{eq:implication_taylor1}
\max_{x\in Z(A)}\|x\|_2^2\geq L
&\quad\Longrightarrow\quad
\max_{y\in Z(B)}\|y\|_p^p\geq L^*,\\
\label{eq:implication_taylor2}
\max_{x\in Z(A)}\|x\|_2^2\leq L-\Delta
&\quad\Longrightarrow\quad
\max_{y\in Z(B)}\|y\|_p^p\leq L^*-\Delta^*.
\end{align}
\end{proposition}
\begin{proof}[Proof sketch]
Write $A=(a_1,\dots,a_q)$ and let $\varepsilon>0$ be a sufficiently small value.
Define
\[
B\coloneqq
\begin{pmatrix}
\mathbf 1_d & \varepsilon a_1 & \dots & \varepsilon a_q\\
\mathbf 1_d &-\varepsilon a_1 & \dots &-\varepsilon a_q
\end{pmatrix}
\in\Q^{2d\times(q+1)}.
\]
The resulting zonotope is $Z(B)=\conv(\{0,\mathbf 1_{2d}\})+\varepsilon (Z(A),-Z(A)).$
Since the maximum is attained at a vertex of the zonotope, we have
\begin{equation}
\label{eq:taylor}
\max_{y\in Z(B)}\|y\|_p^p
=
\max\left\{
\max_{x\in Z(A)}\|\varepsilon (x,-x)\|_p^p,
\max_{x\in Z(A)}\|(\mathbf 1_d+\varepsilon x,\mathbf 1_d-\varepsilon x)\|_p^p
\right\}.
\end{equation}
Using Taylor's theorem, we then show that $\|(\mathbf 1_d+\varepsilon x,\mathbf 1_d-\varepsilon x)\|_p^p$ is approximately equal to $2d+p(p-1)\varepsilon^2\|x\|_2^2$ up to a small fourth-order error.
Then, setting $L^*\coloneqq 2d+p(p-1)\varepsilon^2L-\Delta^*/2$ for a suitable value of $\Delta^*>0$, one can show that the left maximum term inside the maximum in \eqref{eq:taylor} is always smaller than $L^*-\Delta^*$, and that the implications \eqref{eq:implication_taylor1} and \eqref{eq:implication_taylor2} hold.
\end{proof}
Applying \Cref{prop:L2toLpTransfer} to the construction from
\Cref{prop:second_reduction_result} yields the following result.

\begin{theorem}\label{thm:Lp_hardness}
	For every fixed $p\in (1,\infty)\cap \Q$, \textsc{$L^p_p$-Max on Zonotopes} is W[1]-hard with respect to $d$ and not solvable in time $\rho(d) N^{o(d)}$ (where $N$ is the input bit-length) for any computable function $\rho$ assuming the ETH.
\end{theorem}

Since the threshold in \textsc{$L_p$-Max on Zonotopes} must be rational, one cannot simply replace the threshold $L^*$ from \Cref{thm:Lp_hardness} with $(L^*)^{1/p}$ to transfer the hardness of \textsc{$L_p^p$-Max on Zonotopes} to \textsc{$L_p$-Max on Zonotopes}.
However, since the construction creates a positive gap in the norm between yes- and no-instances, we can still transfer the hardness and obtain the following corollary.

\begin{corollary}\label{thm:Lp_norm_hardness}
	For every fixed $p\in (1,\infty)\cap \Q$, \textsc{$L_p$-Max on Zonotopes} and \textsc{Two-Layer ICNN $L_p$-Lipschitz Constant} are W[1]-hard with respect to $d$ and not solvable in time $\rho(d) N^{o(d)}$ (where $N$ is the input bit-length) for any computable function $\rho$ assuming the ETH.
\end{corollary}
\begin{proof}[Proof sketch]
Write $p=a/b$ as an irreducible fraction with $a,b\in \N_{\geq 1}$.
Write $A=L^*$ and $B=L^*-\Delta^*$ for the thresholds from \Cref{thm:Lp_hardness}.
We can compute a rational number $T$ of polynomial encoding size with $B^{1/p}<T\leq A^{1/p}$ in polynomial time by starting with a sufficiently large interval $[0, U]$ ($U\in \Q_{>0}$) and iteratively discarding the half of the interval that does not contain $A^{1/p}$ by checking whether the rational midpoint $q$ satisfies $q\leq A^{1/p}=A^{b/a}$ (by comparing $q^a$ and $A^b$).
After a linear number of steps, we obtain $T$ as the rational lower endpoint of the interval containing $A^{1/p}$.
Thus, the reduction from \Cref{thm:Lp_hardness} applies to \textsc{$L_p$-Max on Zonotopes} using $T$ as the threshold, and we obtain the same hardness results.
The equivalence described in \Cref{sec:icnn_lipschitz} implies the same hardness result for \textsc{Two-Layer ICNN $L_p$-Lipschitz Constant}.
\end{proof}
We note that hardness for two-layer ICNNs implies hardness for deeper ICNNs with $\ell\geq 2$ layers by concatenating the two-layer ICNN with additional layers computing the identity map.

\section{Conclusion}
We prove W[1]-hardness with respect to the dimension $d$ of exact $L_p$-norm maximization over generator-represented zonotopes in $\R^d$ for every fixed $p\in (1,\infty)\cap \Q$.
Assuming the ETH, the problem is not solvable in time $\rho(d)N^{o(d)}$ for any computable function $\rho$.
The same hardness result holds for maximizing the $p$-th power of the $L_p$-norms for $p\in (1,\infty)\cap \Q$.
As a consequence of the duality between zonotopes and two-layer ReLU ICNNs, these results imply the same parameterized hardness and ETH lower bounds for computing the $L_p$-Lipschitz constants for $p\in (1,\infty)\cap \Q$ of two-layer ReLU ICNNs.
Therefore, restricting a ReLU network to be input-convex does not make the problem tractable, in contrast to the $L_1$- and $L_\infty$-norms.
Moreover, the running time lower bounds imply that simple brute-force enumeration algorithms based on enumerating zonotope vertices or linear regions of the neural network are essentially optimal under the ETH.

Since restricting to input-convex ReLU networks does not make the $L_p$-Lipschitz constant tractable, it would be interesting to identify other special cases, additional restrictions on the network architecture or weights, or parameter combinations that lead to tractability.
Also, it would be interesting to understand whether the gadgets and techniques (in particular, \Cref{lem:interpolation}) presented here are useful for other geometric problems (especially for hardness reductions involving zonotopes).
Finally, it remains open whether the problems are also contained in W[1], which would settle the parameterized complexity status completely.

\section{AI use statement and research process with LLMs}
\label{sec:AIprocess}
We used LLMs to provide critical ingredients for proving mathematical claims, assist in writing proofs, search for related literature, assist in creating TikZ figures, and polish writing.
As discussed in more detail below, initial versions of the reductions in both \Cref{sec:L2hardness} and \Cref{sec:reduction-explicit} were found by LLMs and then improved through a combination of human insight and AI assistance.
While the constructions and proofs were obtained with AI assistance, we wrote the presentation of all proofs ourselves, focusing on conveying intuition and the underlying geometric ideas.
We have reviewed all AI-assisted work and take full responsibility for the content of this work.

We conclude with a more detailed description of our research process involving LLMs and some opinions on mathematical research in the age of AI.

As we pointed out in the introduction, the resolution of the open problem presented in this paper has also been achieved by others around the same time.
The homepage with the agentic AI pipeline (\Cref{fn:ai-pipeline}) with the result from July 2026 is very explicit about AI use, and we also state clearly that our two reductions were independently found using LLMs: initial versions of the reductions in both \Cref{sec:L2hardness} and \Cref{sec:reduction-explicit} were found independently in May 2026 by different subsets of the authors. Both groups spent significant effort on independent versions throughout June 2026, before all authors joined forces only in July 2026.
While the advent of AI-generated results is pushing the race for fast publications to an extreme, we decided to invest the additional time of a `human layer' aiming at a high-quality and easy-to-verify exposition before making the result public.

It may not be surprising that this result was found independently by different groups with the use of LLMs. 
Given a list of open problems, like the one from COLT 2025 \citep{pmlr-v291-froese25b}, no expertise on the respective subject is actually required to prompt an LLM for a solution of the problem.
Furthermore, this is one of the problems where AI seems promising as finding reductions of this kind often involves only a few tricks that might have appeared scattered in the literature before and some tedious adaptations of actual symbolic expressions.
This may be inherently easier for machines than for humans.
Note that our reductions were found with rather basic prompts and no elaborate agentic pipeline.

Already in the process of prompting, one aspect of the human ingredient was to push for strengthenings of the reduction, like lowering the dimension of the geometric gadgets from quadratic to linear.
Still, from a human point of view, the LLM proofs had gaps, flaws and were overly technical.
Furthermore, the writing was full of pseudo-intuitive expressions, distracting from the actual argument. 
A mathematical proof should not only be correct but additionally convincing and conceptually clear.
Despite significant attempts, further prompting of LLMs led neither to a gain of intuition or clarity, nor to the identification of some crucial previous work. 
Using our human expertise, we identified that various ideas in the LLM-generated proofs had appeared before in the literature.
Working through the LLM output led to an understanding of the geometric intuition and the extraction of conceptual insights as presented above. 
In particular, we identified \Cref{lem:cone_structure} as an interesting geometric observation that was hidden in the elaborations leading to the desired result.

Summarizing, we made an effort to create digested mathematics instead of just producing proofs, in particular by providing intuition and extracting conceptual insights.
There is a difference between merely correct proofs and actually understandable proofs.
While the result is interesting in its own right, ingredients of proofs have always been a crucial contribution and this should not be hidden in the stream of machine-generated arguments.
That was also a reason to include two reductions into this paper as they provide different conceptual insights. 
With this, we hope to contribute to the development of mathematics in the age of AI, focusing on high-quality exposition, and drawing inspiration from~\citep{tao2026mathematicsageai}.

\bibliography{iclr/references}
\bibliographystyle{iclr2027_conference}

\appendix
\crefalias{section}{appendix}

\section{Proofs}\label{app:proofs}
\subsection{Proof of \texorpdfstring{\Cref{lem:l2max_binaryformulation}}{2.1}}
\begin{proof}
Since $\norm{\cdot}_p^p$ is convex, it attains its maximum over the polytope $Z(G)$ at a vertex.
By~\citep{Ziegler95}, the vertices are vectors of the form $\sum_{i=1}^n \sigma_i g_i\in Z(G)$ for $\sigma\in \{0,1\}^n$, and each such vector is also contained in $Z(G)$.
For every vertex $v$, let $\sigma(v)\in \{0,1\}^n$ be its corresponding binary vector.
It follows that 
	\[
	\max_{z\in Z(G)} \norm{z}_p^p
	=
	\max_{v\in V(Z(G))}\norm{\sum_{i=1}^n \sigma(v)_i g_i}_p^p
	=
	\max_{\sigma\in\{0,1\}^n}\norm{\sum_{i=1}^n \sigma_i g_i}_p^p.
	\qedhere
	\]
\end{proof}

\subsection{Proof of \texorpdfstring{\Cref{lem:cone_structure}}{4.1}}
\begin{proof}
    \textbf{$C(S)$ is finitely generated.}
    Let $U \coloneqq \text{span}(S)$ and $U^\perp=\{w\in \R^r: w^\top u = 0 \text{ for all }u\in U\}$.
    Consider the hyperplane arrangement $\mathcal{H}\coloneqq\left(\{w\in U:s_m^\top w=0\}\right)_{m\in [M]}$.
    On each cone $Q$ of~$\mathcal{H}$, $\Phi_S$ is linear since the signs of $s_m^\top w$ are fixed for all $w\in Q$.
    Moreover, each cone~$Q$ of $\mathcal{H}$ is pointed: if $w\neq 0$ with $\pm w\in Q$, then $w^\top s_m=0$ for all $m\in [M]$ and thus $w\in U^\perp$, so $w\notin U$, which gives a contradiction.
    Thus, each cone $Q$ is generated by its extreme rays.
    
    Let $\R w_1, \dots, \R w_q$ be all one-dimensional faces of~$\mathcal{H}$.
    We can enumerate them in polynomial time, using the fact that every such face is the intersection of $\dim(U)-1$ linearly independent hyperplanes from $\mathcal H$.
    Thus, by enumerating all subsets $I\subseteq [M]$ of $\dim(U)-1$ linearly independent hyperplanes and computing a nonzero rational vector $w_I$ spanning $\{w\in U:s_i^\top w=0 \text{ for }i\in I\}$, we obtain all faces of dimension one.
    In particular, there are at most $q\leq M^{\mathcal{O}(r)}$ such faces.

    Now, fix a $w\in \R^r$.
    We can write it as $w=y+z$ with $y\in U$ and $z\in U^\perp$.
    In particular, $y$ lies in a pointed cone $Q$ of $\mathcal{H}$ and is therefore a conic combination of extreme rays of $Q$, so
    \[
    w=z+\sum_{t=1}^q\lambda_t w_t,
    \qquad
    \text{for }z\in U^\perp, \lambda_t\in \R,
    \]
    since the generators of the extreme rays of $Q$ form a subset of the generators of the one-dimensional faces of $\mathcal H$ and their negative copies.
    Since $\Phi_S(z)=0$, $\Phi(\lambda_t w_t)=|\lambda_t|\Phi(w_t)$ for all $\lambda_t\in \R$, and since $\Phi_S$ is linear on every cone $Q$ of $\mathcal{H}$, we have
    \[
    \Phi_S(w)=\sum_{t=1}^q|\lambda_t| \Phi_S(w_t),
    \]
    which proves that $C(S)=\cone\{\Phi_S(w_t):t\in [q]\}$ is finitely generated.
    We now claim that $\mathcal R=\{\Phi_S(w_t):t\in [q]\}$ is actually exactly the set of extreme rays of $C(S)$.
    
    To see this, fix a generator $w_I$ of the one-dimensional face $\R w_I = \bigcap_{i\in I}\{w\in U:s_i^\top w=0\}$ and suppose that $\Phi_S(w_I)=y+z$ with $y,z\in C(S)$.
    W.l.o.g.\ let $y=\sum_{j\in J} \lambda_j \Phi_S(w_j)$ with $\lambda_j>0$ and $J\subseteq [q]$.
    We have $\Phi_S(w_I)_i=|s_i^\top w_I|=0$ for all $i\in I$, and therefore also $y_i=\sum_{j\in J}\lambda_j |s_i^\top w_j|=0$ for all $i\in I$.
    Since $\lambda_j>0$ for all $j\in J$, we have $s_i^\top w_j=0$ for all $j\in J$ and $i\in I$.
    Then, $w_j\in \bigcap_{i\in I}\{w\in U:s_i^\top w=0\}=\R w_I$, so $w_j=\alpha_jw_I$ for some $\alpha_j\in \R$.
    Then, $\Phi_S(w_j)=|\alpha_j|\Phi_S(w_I)$ and thus
    $
    y=\Phi_S(w_I)\left(\sum_{j\in J} \lambda_j |\alpha_j|\right),
    $
    that is, $y$ is a scaled version of $\Phi_S(w_I)$.
    Analogously, the same holds for $z$, and therefore $\Phi_S(w_I)$ is an extreme ray of $C(S)$.
    
    \textbf{$\mathbf{1}$ lies in the interior of $C(S)$.}
    Suppose that $\mathbf{1}\notin \operatorname{int}C(S)$.
    Then, a variant of Farkas' lemma gives a nonzero $\alpha\in \R^M$ with
    \[
    \alpha^\top x\geq 0
    \;\text{ for all }x\in C(S),
    \qquad
    \alpha^\top \mathbf{1}\leq 0.
    \]
    Define the continuous function $f\colon \R^r\to\R$,  $f(w)\coloneq\alpha^\top \Phi_S(w)=\sum_{m=1}^M \alpha_m|w^\top s_m|$ and observe that $f(w)\geq 0$ (since $\Phi_S(w)\in C(S)$) for all $w\in \R^r$.
    Since $\alpha\neq 0$, there is at least one $\alpha_m\neq 0$, and the nonzero summand $\alpha_m|w^\top s_m|$ in $f$ implies that $f$ is nonlinear in an open set around any point lying on the hyperplane $H_m\coloneqq\{w\in \R^r:w^\top s_m=0\}$ since the vectors in $S$ are pairwise linearly independent and therefore induce different hyperplanes $\{w\in \R^r:w^\top s_i=0\}\neq H_m$ for all $i\neq m$, so other terms $\alpha_i|w^\top s_i|$ cannot cancel out the nonlinearity of $f$ along $H_m$.
    Hence, $f$ is nonzero.
    
    Let $S_U = \{u \in U:\|u\|_2 = 1\}\subseteq\R^r$.
    Further, let $\sigma$ be the rotation-invariant probability measure with full support on $S_U$.
    Then, there exists a $c_U > 0$ such that for every $s\in U$,
    \[
    \int_{S_U}|w^\top s|\,\text{d}\sigma(w)
    =\|s\|_2\int_{S_U}|w^\top s|/\|s\|_2\,\text{d}\sigma(w)
    =c_U\|s\|_2.
    \]
    To see this, observe that the rotation-invariance of $\sigma$ implies that $\int_{S_U}|w^\top s|/\|s\|_2\,\text{d}\sigma(w)$ does not depend on the choice of the unit-vector $s/\|s\|_2$.
    After identifying $U$ with $\R^{r'}$ (with $r'=\dim(U)$) and choosing $e_1$ as the unit vector, it follows that the integral equals the expected value $\mathbb{E}_{x \sim \mathbb{S}^{r'-1}}[|x_1|]=c_U>0$ of the absolute value of the first coordinate of a point uniformly distributed on the unit sphere $\mathbb{S}^{r'-1}$.
    
    Now, since all $s\in S$ have equal Euclidean norm $\gamma\coloneq \|s_1\|_2>0$, we have
    \[
    \int_{S_U}f(w)\,\text{d}\sigma(w)
    =
    \sum_{m=1}^M \alpha_m\int_{S_U}|w^\top s_m|\,\text{d}\sigma(w)
    =\alpha^\top (c_U\gamma,\dots,c_U\gamma)
    =c_U\gamma\,\alpha^\top \mathbf{1}
    \leq 0.
    \]
    Since $f$ is nonnegative, the integral must be equal to zero.
    This implies, since $f$ is continuous, nonnegative, and since $\sigma$ has full support, that $f$ is identically zero on~$S_U$.
    Since $f$ is positively homogeneous, $f$ is identically zero on $U$.
    For every $w\in \R^r$, if $w_U$ denotes the projection of $w$ onto $U$, we have $w^\top s_m=w_U^\top s_m$ for all $m\in [M]$, so $f$ is identically zero on $\R^r$, contradicting that $f$ is nonzero.
    Thus, there is no nonzero $\alpha$ satisfying the above equations and hence $\mathbf{1}\in \operatorname{int}C(S)$.
\end{proof}

\subsection{Proof of \texorpdfstring{\Cref{prop:second_reduction_result}}{4.3}}
We first construct the zonotope whose support function is the node term $F_V$, which enforces the choice of one node per color class.

\begin{lemma}\label{lem:node_selection}
    Given an instance $(G,k)$ of \textsc{Multicolored Clique}, one can construct in polynomial time a zonotope $Z\subset Y$ and define a finite set $\mathcal{T}\subseteq Y$ of unit vectors, where each point encodes one node per color, with
    \[
    h_Z(y)=1-\frac 12\operatorname{dist}(y,\mathcal{T})^2
    \qquad
    \text{for every } y\in Y\text{ with }\|y\|_2=1.
    \]
\end{lemma}
\begin{proof}
    We choose positive rational numbers $\omega_1, \dots, \omega_k$ with $\sum_{i} \omega_i^2 = 1$ of polynomial bit-length.
    Specifically,
    \[
        \omega_i = \frac{4k}{4k^2 + k - 1}\;\text{ for }i \in [k-1], \qquad \omega_k = \frac{4k^2 - k + 1}{4k^2 + k - 1}.
    \]
    For every node $v_{i,a}\in V_i$ with $a\in [n_i]$, we define the rational unit vector
    \[
    u_{i, a} \coloneq \frac{(1-a^2 , 2a)}{1+a^2}\in \Q^2.
    \]
    Let $P_i\coloneq\conv\{\pm u_{i,a}: a\in [n_i]\}\subset \R^2$.
    List the vertices of $P_i$ in cyclic order as $u_{i,1},\dots, u_{i,2n_i}$ with $u_{i,r+n_i}\coloneqq -u_{i,r}$ and define
    \[
    g_{i,r}\coloneq \frac{u_{i,r+1}-u_{i,r}}{2}\quad
    \text{for } r\in [n_i].
    \]
    Then we have that
    \[
    P_i=\sum_{r=1}^{n_i}\conv\left(\{-g_{i,r},g_{i,r}\}\right),
    \qquad
    h_{P_i}(y)=\sum_{r=1}^{n_i}\left|g_{i,r}^\top y\right|.
    \]
    We now embed each zonotope $P_i$ into $Y_i$ and define
    \[
    Z\coloneqq\sum_{i=1}^k\rho_i(\omega_i P_i),
    \qquad
    F_V(y)\coloneq h_Z(y)=\sum_{i=1}^k \omega_i h_{P_i}(y_i).
    \]
    Finally, let
    \[
    \mathcal{T}\coloneq 
    \{(\sigma_1 \omega_1 u_{1,a_1},\dots, \sigma_k \omega_k u_{k,a_k}): a_i\in [n_i], \sigma\in \{\pm 1\}^k\}\subset Y.
    \]
    Every $t\in \mathcal{T}$ is a unit vector.
    For any unit vector $y=(y_1,\dots, y_k)\in Y$, we have
    \begin{align*}
        \dist(y,\mathcal{T})^2
        &=\sum_{i=1}^k \min_{a,\sigma}\|y_i-\sigma \omega_i u_{i,a}\|_2^2\\
        &=\sum_{i=1}^k\left(\|y_i\|_2^2 + \omega_i^2 -2\omega_i\max_{a,\sigma}(\sigma  u_{i,a})^\top y_i\right)\\
        &=2-2\sum_{i=1}^k\omega_i h_{P_i}(y_i)
        =2-2F_V(y),
    \end{align*}
    which is equivalent to $h_Z(y)=1-\frac 12\operatorname{dist}(y,\mathcal{T})^2$.
\end{proof}
Now, we construct the symmetric zonotope whose support function is the edge term $E$ that encodes the edges between the nodes encoded by a point of $\mathcal{T}$.
\begin{lemma}\label{lem:edge_term}
Given an instance $(G,k)$ of \textsc{Multicolored Clique}, one can construct in polynomial time a zonotope $Z_E\subset Y$ with support function $F_E\colon Y\mapsto \R$ and positive rational numbers $E_{\rm cl}$, $\delta_*$ such that, for every vector $t\in \mathcal{T}$ encoding nodes $\{v_{1,a_1},\dots, v_{k,a_k}\}$,
    \[
    F_E(t)
    \begin{cases}
        =E_{\rm cl} & \text{if the nodes $\{v_{1,a_1},\dots, v_{k,a_k}\}$ form a $k$-colored clique,}\\
        \leq E_{\rm cl}-\delta_* & \text{otherwise.}
    \end{cases}
    \]
\end{lemma}
\begin{proof}
    Fix $1\leq i<j\leq k$.
    For $a_i\in [n_i]$, $a_j\in [n_j]$, and $\tau\in \{\pm 1\}$, set
    \[
    s_{a_i, a_j, \tau} \coloneq (\omega_iu_{i, a_i}, \tau \omega_j u_{j, a_j}) \in Y_i \times Y_j\cong \R^4.
    \]
    These $2n_in_j$ vectors are pairwise linearly independent and all have the same Euclidean norm.
    The extra $\tau\in \{\pm 1\}$ is needed later since the $Y_i$ and $Y_j$ components of each $t\in \mathcal T$ can have one of four sign combinations $\{\pm1\}\times \{\pm1\}$.
    Define the vector $b\in\{0,1\}^{2n_in_j}$ with corresponding entries
    \[
    b_{a_i,a_j,\tau}\coloneq
    \begin{cases}
        1& \text{ if }\{v_{i,a_i},v_{j,a_j}\}\in E(G),\\
        0& \text{ otherwise.}
    \end{cases}
    \]
    Applying \Cref{lem:interpolation} for $r=4$ to $S=\{s_{a, b, \tau}: a_i\in [n_i], a_j\in [n_j],\tau\in\{\pm1\}\}$ and $b$ yields rational vectors $g_{ij,1},\dots, g_{ij,q_{ij}}\in Y_i\times Y_j$ with $q_{ij}= 2n_in_j$, and a rational number $\delta_{ij}>0$.
    Define $H_{ij}\colon Y_i\times Y_j\to \R$ as
    \[
    H_{ij}(x_i,x_j)
    \coloneq
    \sum_{t=1}^{q_{ij}}\left|g_{ij,t}^\top (x_i,x_j)\right|.
    \]
    Then
    \[
    H_{ij}(s_{a_i,a_j,\tau})=1+\delta_{ij}b_{a_i,a_j,\tau}
    \qquad
    \text{for }a_i\in [n_i], a_j\in [n_j],\tau\in\{\pm1\}.
    \]
    We now embed every vector $g_{ij,t}$ into $Y$ using $\rho_{ij}$ and define $F_E\colon Y \to \R$ as
    \[
    F_E(y)\coloneq \sum_{1\leq i<j\leq k}H_{ij}(y_i,y_j).
    \]
    Thus, $F_E$ is the support function of the centrally symmetric zonotope $Z_E$ generated by all the projected $\pm g_{ij,t}$.
    Set
    \[
    \delta_* \coloneqq \min_{1\le i < j \le k} \delta_{ij} > 0,
    \qquad
    E_{\rm cl} \coloneqq \sum_{1 \le i < j \le k} (1 + \delta_{ij}).
    \]
    Now, consider a $t=(\sigma_1\omega_1 u_{1,a_1},\dots, \sigma_k\omega_k u_{k,a_k})\in \mathcal{T}$.
    For a pair $i < j$, set $\tau_{ij}\coloneqq\sigma_i\sigma_j$.
    Since $H_{ij}$ is even,
    \[
    H_{ij}(\sigma_i\omega_iu_{i,a_i},\sigma_j\omega_ju_{j,a_j})
    =H_{ij}(s_{a_i,a_j,\tau_{ij}})
    =1+\delta_{ij}b_{a_i,a_j,\tau_{ij}}
    \]
    and thus
    \[
    F_E(t)
    \begin{cases}
        =E_{\rm cl} & \text{if the chosen nodes form a clique,}\\
        \leq E_{\rm cl}-\delta_* & \text{otherwise.}
    \end{cases}
    \qedhere
    \]
\end{proof}

We are now ready to give the reduction that combines the node and edge terms.

\begin{proof}[Proof of \Cref{prop:second_reduction_result}]
    We now combine the node and edge terms.
    We use the objects and notation from the proofs of \Cref{lem:node_selection} and \Cref{lem:edge_term}.
    Note that $\Lambda\coloneq 1+\sum_{1 \le i < j \le k}\sum_{t=1}^{q_{ij}}\|g_{ij,t}\|_1$ is an upper bound on the $L_2$-Lipschitz constant of $F_E$, since for $y,y'\in Y$,
    \begin{align*}
    |F_E(y)-F_E(y')|
    &\leq
    \sum_{1 \le i < j \le k}\sum_{t=1}^{q_{ij}}\left| |g_{ij,t}^\top(y_i,y_j)|-|g_{ij,t}^\top(y'_i,y'_j)|\right|\\
    &\leq \sum_{1 \le i < j \le k}\sum_{t=1}^{q_{ij}}\left|g_{ij,t}^\top(y_i-y'_i,y_j-y'_j)\right|\\
    &\leq \left(\sum_{1 \le i < j \le k}\sum_{t=1}^{q_{ij}}\|g_{ij,t}\|_2\right)\|y-y'\|_2\\
    &\leq \left(\sum_{1 \le i < j \le k}\sum_{t=1}^{q_{ij}}\|g_{ij,t}\|_1\right)\|y-y'\|_2\\
    &<\Lambda\|y-y'\|_2,
    \end{align*}
    where the second inequality is the reverse triangle inequality, the third is Cauchy-Schwarz, and the fourth uses $\|g_{ij,t}\|_2\leq \|g_{ij,t}\|_1$.
    We now set
    \[
    \eta\coloneq \frac{\delta_*}{8\Lambda^2}
    \quad
    \text{and}
    \quad
    F(y)\coloneq F_V(y)+\eta F_E(y).
    \]
    If $G$ has a $k$-colored clique $\{v_{1,a_1},\dots, v_{k,a_k}\}$, then setting $t\coloneqq(\omega_1 u_{1,a_1},\dots, \omega_k u_{k,a_k})\in \mathcal{T}$ gives, by \Cref{lem:edge_term},
    \[
    F(t)=1+\eta E_{\rm cl}.
    \]
    Now suppose that $G$ has no $k$-colored clique.
    Let $y\in Y$ be a unit vector, choose $t\in \mathcal{T}$ with $\|y-t\|_2=\dist(y,\mathcal{T})$ and let $D\coloneqq\|y-t\|_2$.
    By \Cref{lem:node_selection}, $F_V(y)=1-D^2/2$ and the $L_2$-Lipschitz constant bound for $E$ gives
    \[
    F_E(y)
    \leq F_E(t)+\Lambda D
    \leq E_{\rm cl}-\delta_*+\Lambda D.
    \]
    Therefore, using \Cref{lem:edge_term} and $\eta\Lambda^2=\delta_*/8$,
    \begin{align*}
    F(y)
    &\leq 1-D^2/2+\eta(E_{\rm cl}-\delta_*+\Lambda D)\\
    &=1+\eta(E_{\rm cl}-\delta_*)-(D-\eta \Lambda)^2/2+\eta^2\Lambda^2/2\\
    &\leq 1+\eta(E_{\rm cl}-\delta_*)+\eta\delta_*/16\\
    &<1+\eta(E_{\rm cl}-\delta_*/2).
    \end{align*}
    Conceptually, moving away from $\mathcal{T}$ causes a gain of at most $\eta\delta_*/16$, which is strictly smaller than the loss of $\eta\delta_*$ that occurs since $G$ does not have a $k$-colored clique.
    Now, with
    \[
    L\coloneqq 1+\eta E_{\rm cl},
    \qquad
    \Delta\coloneqq\eta \delta_*/2,
    \]
    we have
    \[
    \max_{\|y\|_2=1}F(y)
    \begin{cases}
        \geq L & \text{if $(G,k)$ is a yes-instance,}\\
        \leq L-\Delta & \text{otherwise.}
    \end{cases}
    \]
    By construction, $F$ is the support function of a centrally symmetric zonotope with generators $\pm\omega_i \rho_i(g_{i,r})$ for $i\in [k],\,r\in [n_i]$ and $\pm\eta \rho_{ij}(g_{ij,t})$ for $1\leq i<j\leq k,\,t\in [q_{ij}]$.
    Let $A$ be the matrix with these vectors as columns.
    Then, $F=h_{Z(A)}$, which implies
    \[
    \max_{z\in Z(A)}\|z\|_2
    =
    \max_{\|y\|_2=1}F(y).
    \]
    Then
    \[
    \max_{z\in Z(A)}\|z\|_2
    \begin{cases}
    \geq L & \text{if $(G,k)$ is a yes-instance,}\\
    \leq L-\Delta & \text{otherwise.}
    \end{cases}
    \]
    Finally, observe that all numbers can be computed in polynomial time and have polynomial encoding size in the input.
    Thus, the whole construction is a polynomial-time parameterized reduction from \textsc{Multicolored Clique} to \textsc{$L_2$-Max on Zonotopes} with dimension $2k$.
\end{proof}

\subsection{Proof of \texorpdfstring{\Cref{prop:L2toLpTransfer}}{5.1}}
The geometric idea is to antisymmetrically embed the original zonotope~$Z(A)$ by $x\mapsto(x,-x)$ into the tangent space of the $L_p$-sphere in $\R^{2d}$ at the all-ones vector $b=\mathbf 1_{2d}$, that is, the linear space orthogonal to $b$. After scaling this copy by a sufficiently small factor $\varepsilon$, we add $b$ as an additional generator. The resulting zonotope is the prism
$
\conv(\{0,b\})+\varepsilon JZ(A),
$ where $
Jx\coloneq(x,-x)$,
whose two parallel faces are the zonotopes $\varepsilon JZ(A)$ and $b+\varepsilon JZ(A)$, where $+$ denotes the Minkowski sum. The latter is a translated copy of the former zonotope lying in the affine tangent hyperplane at $b$. For sufficiently small $\varepsilon$, this translated face dominates the $L_p^p$-objective, whereas the untranslated face remains close to the origin.

Since the translated zonotope lies in the affine tangent hyperplane, the first-order term of the Taylor expansion of the $L_p^p$-objective vanishes along it. The leading term of the Taylor expansion is therefore the second order term. Moreover, the Hessian of the $L_p^p$-objective at $b$ is $p(p-1)$ times the identity matrix, so this second-order term is proportional to the $L_2^2$-objective. By choosing~$\varepsilon$ sufficiently small, the higher-order terms of the Taylor expansion become smaller than the gap in the original $L_2^2$-maximization instance. Consequently, $L_p^p$-maximization over the new zonotope reproduces $L_2^2$-maximization over the original one.

For $\varepsilon>0$, define the affine embedding
$\Phi_\varepsilon\colon \R^d \to \R^{2d}$, $\Phi_\varepsilon(x)\coloneq b+\varepsilon Jx
=(\mathbf 1_d+\varepsilon x,\mathbf 1_d-\varepsilon x)$.
The following lemma formalizes the local Euclideanization and shows that
it applies uniformly to any bounded compact set.

\begin{lemma}\label{lem:taylor_estimate}
Let $p\in(1,\infty)$. For every $x\in\R^d$ and every $\varepsilon>0$
satisfying $\varepsilon\|x\|_\infty\leq1/2$, we have
\[
\left|
\|\Phi_\varepsilon(x)\|_p^p
-
2d
-
p(p-1)\varepsilon^2\|x\|_2^2
\right|
\leq
C_p\varepsilon^4\|x\|_4^4,
\text{ where }
C_p\coloneq p^4 2^{\lceil p\rceil}.
\]
In particular, let $K\subseteq[-R,R]^d$ be a nonempty compact set and
set $M_K\coloneq\max_{x\in K}\|x\|_2^2$.

If
$\varepsilon R\leq1/2$, then
\[
\left|
\max_{x\in K}\|\Phi_\varepsilon(x)\|_p^p
-
2d
-
p(p-1)\varepsilon^2M_K
\right|
\leq
dC_p\varepsilon^4R^4.
\]
\end{lemma}
\begin{proof}
This is a direct application of Taylor's theorem with Lagrange
remainder; see, for example, \citep[Theorem 5.15]{rudin1976principles}.
Define
\[
g_p:[-1/2,1/2]\to\R_{\geq0},
\qquad
t\mapsto(1+t)^p+(1-t)^p.
\]
The function is even, so it suffices to consider $t\in[0,1/2]$.
Let $g_p^{[n]}$ denote the $n$-th derivative of $g_p$. We have
\[
g_p(0)=2,
\qquad
g_p^{[1]}(0)=g_p^{[3]}(0)=0,
\qquad
g_p^{[2]}(0)=2p(p-1),
\]
and
\[
g_p^{[4]}(t)
=
p(p-1)(p-2)(p-3)
\left((1+t)^{p-4}+(1-t)^{p-4}\right).
\]
Taylor's theorem gives a point $\xi\in[0,t]$ such that
\[
g_p(t)-2-p(p-1)t^2
=
\frac{g_p^{[4]}(\xi)}{24}t^4.
\]

Suppose first that $p\in(1,4]$. Since
$1\pm\xi\in[1/2,3/2]$, we have
$(1\pm\xi)^{p-4}\leq2^{4-p}$ and therefore
\[
(1+\xi)^{p-4}+(1-\xi)^{p-4}\leq2^{5-p}.
\]
Combining this with
$|p(p-1)(p-2)(p-3)|\leq2p^3$ yields
\[
\left|g_p(t)-2-p(p-1)t^2\right|
\leq
\frac{2p^3}{24}2^{5-p}t^4
=
\frac{p^32^{3-p}}{3}t^4
\leq C_pt^4.
\]
Now suppose that $p>4$. In this case,
$(1\pm\xi)^{p-4}\leq(3/2)^{p-4}$, and hence
\[
\left|g_p(t)-2-p(p-1)t^2\right|
\leq
\frac{p^4}{12}\left(\frac32\right)^{p-4}t^4
\leq C_pt^4.
\]

We apply this estimate with $t=\varepsilon x_i$ for every $i\in[d]$.
Since $\varepsilon\|x\|_\infty\leq1/2$, all these values belong to
$[-1/2,1/2]$. Consequently,
\[
\begin{aligned}
\|\Phi_\varepsilon(x)\|_p^p
&=
\sum_{i=1}^d
\left((1+\varepsilon x_i)^p+(1-\varepsilon x_i)^p\right)\\
&=
2d+p(p-1)\varepsilon^2\|x\|_2^2+\mathcal R(x),
\end{aligned}
\]
where
$
|\mathcal R(x)|
\leq C_p\varepsilon^4\sum_{i=1}^d x_i^4
=C_p\varepsilon^4\|x\|_4^4.
$
This proves the first statement.

Finally, let $K\subseteq[-R,R]^d$. The first statement gives, uniformly
for all $x\in K$,
\[
\left|
\|\Phi_\varepsilon(x)\|_p^p
-
2d
-
p(p-1)\varepsilon^2\|x\|_2^2
\right|
\leq dC_p\varepsilon^4R^4.
\]
Taking maxima over $x\in K$ preserves this error bound.
To see this, let $g(x) = 2d+p(p-1) \varepsilon^2\|x\|_2^2$ and let $x^* \in \operatorname{argmax}_{x \in K} \|\Phi_\varepsilon(x)\|_p^p$.
Then,
\[
\|\Phi_\varepsilon(x^*)\|_p^p
\leq g(x^*) +dC_p\varepsilon^4R^4
\leq \max_{x\in K} g(x) + dC_p\varepsilon^4R^4
\]
and analogously $ \max_{x\in K} g(x) \leq \|\Phi_\varepsilon(x^*)\|_p^p +dC_p\varepsilon^4R^4$, implying the second statement.
\end{proof}

The preceding lemma applies to any bounded compact set. We now
specialize it to zonotopes, for which the embedding has an
explicit generator representation. This gives a general transfer from
$L_2^2$-maximization with a positive gap to $L_p^p$-maximization.

\begin{proof}[Proof of \Cref{prop:L2toLpTransfer}]
Write $A=(a_1,\dots,a_q)$ and let $\varepsilon>0$ be specified below.
Using the notation $b=\mathbf 1_{2d}$ and $Jx=(x,-x)$ introduced
above, define
\[
B=
\begin{pmatrix}
\mathbf 1_d & \varepsilon a_1 & \dots & \varepsilon a_q\\
\mathbf 1_d &-\varepsilon a_1 & \dots &-\varepsilon a_q
\end{pmatrix}
\in\Q^{2d\times(q+1)}.
\]
The resulting zonotope is
$
Z(B)=\conv(\{0,b\})+\varepsilon JZ(A).
$
Since the maximum is attained at a vertex of the zonotope, we have that
\[
\max_{y\in Z(B)}\|y\|_p^p
=
\max\left\{
\max_{x\in Z(A)}\|\varepsilon Jx\|_p^p,
\max_{x\in Z(A)}\|\Phi_\varepsilon(x)\|_p^p
\right\}.
\]
We refer to these two terms as the maxima over the lower and upper
faces of the prism, respectively.

First, observe that $Z(A)\subseteq[-R,R]^d$, where
$
R\coloneq\max\{1,q\max_{i\in[q],\,j\in[d]}|(a_i)_j|\}.
$
Set $M\coloneq\max_{x\in Z(A)}\|x\|_2^2$. For the upper face, assuming that
$\varepsilon R\leq1/2$, \Cref{lem:taylor_estimate} gives
\[
\left|
\max_{x\in Z(A)}\|\Phi_\varepsilon(x)\|_p^p
-
2d
-
p(p-1)\varepsilon^2M
\right|
\leq
dC_p\varepsilon^4R^4.
\]
We set
\[
\Delta^*
\coloneq
\frac{p(p-1)\varepsilon^2\Delta}{2},
\qquad
L^*
\coloneq
2d+p(p-1)\varepsilon^2L-\frac{\Delta^*}{2}.
\]
We choose $\varepsilon$ sufficiently small that
\begin{equation}\label{eq:transfer1}
dC_p\varepsilon^4R^4\leq\frac{\Delta^*}{2}.
\end{equation}

Suppose first that $M\geq L$. Then
\[
\max_{x\in Z(A)}\|\Phi_\varepsilon(x)\|_p^p
\geq
2d+p(p-1)\varepsilon^2L-\frac{\Delta^*}{2}
=
L^*.
\]
Now suppose that $M\leq L-\Delta$. Then
\begin{align*}
\max_{x\in Z(A)}\|\Phi_\varepsilon(x)\|_p^p
&\leq
2d+p(p-1)\varepsilon^2(L-\Delta)
+dC_p\varepsilon^4R^4\\
&\leq
L^*-\frac{3\Delta^*}{2}+\frac{\Delta^*}{2}
=
L^*-\Delta^*.
\end{align*}

It remains to bound the lower face. For every $x\in Z(A)$, we have
\[
\|\varepsilon Jx\|_p^p
=
2\varepsilon^p\|x\|_p^p
\leq
2d(\varepsilon R)^p
<d,
\]
where the last inequality follows from $\varepsilon R\leq1/2$ and
$p>1$. On the other hand,
\[
L^*-\Delta^*-d
=
d+p(p-1)\varepsilon^2
\left(L-\frac{3\Delta}{4}\right)
>0,
\]
since $L>\Delta$. Hence, the maximum over the lower face is strictly
smaller than $L^*-\Delta^*$.

We now choose a rational $\varepsilon$ satisfying
$\varepsilon R\leq1/2$ and \eqref{eq:transfer1}. Set
\[
\varepsilon
\coloneq
\min\left\{
\frac{1}{2R},
\frac{p(p-1)\Delta}
{2dC_pR^2(1+p(p-1)\Delta)}
\right\}.
\]
The first term ensures that $\varepsilon R\leq1/2$. To verify
\eqref{eq:transfer1}, let $z=p(p-1)\Delta$. Then
\[
\varepsilon^2
\leq
\frac{z^2}{4d^2C_p^2R^4(1+z)^2}
\leq
\frac{z}{4dC_pR^4}
=
\frac{p(p-1)\Delta}{4dC_pR^4},
\]
where we used $(1+z)^2\geq z$ and $dC_p\geq1$. Multiplying by
$dC_p\varepsilon^2R^4$ gives \eqref{eq:transfer1}.

The two face estimates now imply the claimed gap for $Z(B)$. Finally,
since $p$ is fixed, all entries of $B$ and the thresholds
$L^*,\Delta^*$ are rational, can be computed in polynomial time, and
have encoding size polynomial in that of $A,L$, and $\Delta$.
\end{proof}

\subsection{Proof of \texorpdfstring{\Cref{thm:Lp_hardness}}{5.2}}
\begin{proof}
We reduce from \textsc{Multicolored Clique}.
Given an instance $(G,k)$,
let $A_G\in\Q^{d\times q}$ be the generator matrix and let $L>\Delta>0$ be the
thresholds constructed in the proof of \Cref{prop:second_reduction_result}.
Recall that
\[
\max_{x\in Z(A_G)}\|x\|_2^2
\begin{cases}
\geq L^2 & \text{if $(G,k)$ is a yes-instance,}\\
\leq (L-\Delta)^2 & \text{if $(G,k)$ is a no-instance.}
\end{cases}
\]
We can therefore apply \Cref{prop:L2toLpTransfer} to $A_G$, $L'=L^2$, and $\Delta'=L^2-(L-\Delta)^2$.
This gives a matrix $B_G\in\Q^{2d\times(q+1)}$ and rational thresholds $L^*,\Delta^*>0$ such that
\[
\max_{y\in Z(B_G)}\|y\|_p^p
\begin{cases}
\geq L^* & \text{if $(G,k)$ is a yes-instance,}\\
\leq L^*-\Delta^* & \text{if $(G,k)$ is a no-instance.}
\end{cases}
\]

For fixed $p$, all constructed rational numbers have encoding size
polynomial in the encoding size of $(G,k)$. 
    The dimension is $D=4k\in \mathcal O(k)$.
    Thus, the reduction is polynomial and proves W[1]-hardness with respect to $d$.
	Finally, an algorithm with running time $\rho(D) N^{o(D)}$ would again contradict the ETH lower bound for \textsc{Multicolored Clique}.
\end{proof}

\subsection{Proof of \texorpdfstring{\Cref{thm:Lp_norm_hardness}}{5.3}}
We first prove a simple preparatory lemma.
\begin{lemma}\label{lem:bisection}
Fix $p\in (1,\infty)\cap \Q$.
Given rational numbers $A>B>0$, it is possible to compute in polynomial time a rational number $T$ of polynomial encoding size with
\[
B^{1/p}<T\leq A^{1/p}.
\]
\end{lemma}
\begin{proof}
    Let $\alpha=A^{1/p}$ and $\beta=B^{1/p}$, and set $U=\max\{1,A\}$.
    The mean value theorem applied to $x\mapsto x^{1/p}$ and the interval $[B,A]$ gives a $\xi\in (B,A)$ such that
    \[
    \alpha-\beta
    =\frac{A-B}{p\xi^{1-1/p}}
    \geq \frac{A-B}{pA^{1-1/p}}
    \geq \frac{A-B}{pU}
    \eqcolon \eta
    >0.
    \]
    The number $\eta$ is rational and has encoding size polynomial in the encoding length of $A$ and $B$.

    Write $p=a/b$ as an irreducible fraction with $a,b\in \N_{\geq 1}$.
    We perform binary search for $\alpha$ on the interval $[0,U]$.
    For every rational interval midpoint $q\geq 0$, the comparison $q\leq \alpha=A^{b/a}$ can be decided exactly by comparing $q^a$ and $A^b$.
    We stop when the interval has length less than $\eta/2$, and let $T$ be the lower endpoint.
    Then $T\leq \alpha$ and
    \[
    T>\alpha-\eta/2>\alpha-\eta >\beta.
    \]
    The binary search takes $k\in \mathcal{O}(\log(U/\eta))$ steps, so $k$ is linear in the encoding size of the input numbers $A$ and $B$.
    Hence, the encoding size of $T$ is linear in the encoding size of the input.
\end{proof}

\begin{proof}[Proof of \Cref{thm:Lp_norm_hardness}]
    We apply \Cref{lem:bisection} to the thresholds $A=L^*$ and $B=L^*-\Delta^*$ from the proof of \Cref{thm:Lp_hardness} and compute in polynomial time a rational number of polynomial encoding size $T$ with $B^{1/p}<T\leq A^{1/p}$.
    In a yes-instance, we have $\max_{z\in Z(B_G)} \norm{z}_p\geq (L^*)^{1/p}\geq T$ and in a no-instance, we have $\max_{z\in Z(B_G)} \norm{z}_p\leq (L^*-\Delta^*)^{1/p}<T$.
    Thus, the reduction from the proof of \Cref{thm:Lp_hardness} applies analogously to \textsc{$L_p$-Max on Zonotopes}, and we obtain the same hardness results.
    The equivalence between \textsc{$L_p$-Max on Zonotopes} and \textsc{Two-Layer ICNN $L_q$-Lipschitz Constant} described in \Cref{sec:icnn_lipschitz} implies the same hardness result for \textsc{Two-Layer ICNN $L_q$-Lipschitz Constant} for all fixed $q\in (1,\infty)\cap \Q$.
\end{proof}

\section{Explicit Reduction with \texorpdfstring{$2k+1$}{2k+1} Dimensions}
\label{sec:reduction-explicit}

In this section, we give an elementary explicit reduction with $2k+1$ dimensions.
Compared to the reduction in \Cref{sec:L2hardness} which obtains the generators of the zonotope non-explicitly from the solution of linear programs, this reduction explicitly constructs the generators of a zonotope and relies only on elementary techniques.
This comes at the cost of using one additional coordinate to control the choice of edges, resulting in $2k+1$ dimensions compared to $2k$ dimensions.

We first describe a gadget that we use in the reduction.
It constructs rational vectors $u_0,\dots, u_m$ on the $L_2$-norm unit circle, except for $u_0=0$, and difference vectors $w_i=u_i-u_{i-1}$ such that for a binary vector $\sigma\in\{0,1\}^m$, the sum $\sum_{i=1}^m\sigma_i w_i$ lies on the unit circle if $\sigma$ is nonzero and decreasing, and lies strictly inside the unit circle otherwise (see \Cref{fig:gadget} for an illustration).
\begin{figure}
	\centering
    \resizebox{0.65\linewidth}{!}{
	\includegraphics[page=2, width=\textwidth]{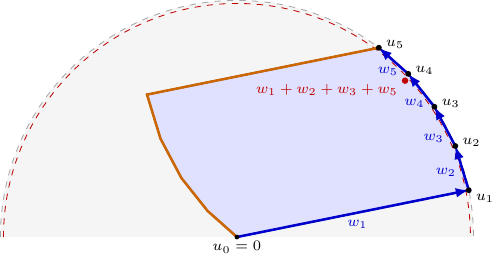}
    }
	\caption{
		The vectors from \Cref{lem:gadget} for $m=5$.
		The point $w_1+w_2+w_3+w_5$ (red) corresponds to the not decreasing binary vector $\sigma=(1,1,1,0,1)$ and lies in a circle with radius $\sqrt{1-\delta_m}$ (red).
 	}
	\label{fig:gadget}
\end{figure}
\begin{lemma}\label{lem:gadget}
	Let $m\in\mathbb{N}_{\geq 1}$, let $u_0=(0,0)$, and, for every
	$i\in[m]$, define
	\[
	t_i\coloneqq\frac{i}{2m},
	\qquad
	u_i\coloneqq
	\left(
	\frac{1-t_i^2}{1+t_i^2},
	\frac{2t_i}{1+t_i^2}
	\right),
	\qquad
	w_i\coloneqq u_i-u_{i-1}.
	\]
	For every nonzero decreasing $\sigma\in\{0,1\}^m$, we have $\left\|\sum_{i=1}^m \sigma_i w_i\right\|_2^2=1$.
    For every $\sigma\in\{0,1\}^m$ that is either zero or not decreasing, we have
	\[
	\left\|\sum_{i=1}^m \sigma_iw_i\right\|_2^2
	\leq 1-\delta_m,
	\qquad
	\delta_m:=\frac{16}{25m^2}.
	\]
	
	Moreover, for $m\geq 2$, we have
	$
	\min_{i,j\in [m],\,i\neq j}
	\|u_i-u_j\|_2^2
	\geq
	\delta_m
	$.
\end{lemma}
\begin{proof}
The statement is trivial for $m=1$, and the second claim is trivial for $\sigma=0$.
Let $\sigma\in\{0,1\}^m$ be nonzero and decreasing.
Then there is a $k\in[m]$ such that $\sigma_i=1$ for all $i\leq k$ and $\sigma_i=0$ for all $i>k$.
The sum telescopes, and hence
\[
\left\|\sum_{i=1}^m\sigma_iw_i\right\|_2^2
=
\left\|\sum_{i=1}^k(u_i-u_{i-1})\right\|_2^2
=
\left\|u_k\right\|_2^2
=1,
\]
which proves the first statement.
For all $i,j\in [m],\, i\neq j$, a direct calculation gives
\begin{equation}\label{eq:gadgetcalc}
u_j-u_i=\dfrac{2(t_j-t_i)}{(1+t_j^2)(1+t_i^2)}
\left(
-(t_j+t_i),
1-t_jt_i
\right)
\quad
\text{and}
\quad
\|u_j-u_i\|_2^2
=\dfrac{4(t_j-t_i)^2}{(1+t_j^2)(1+t_i^2)}.
\end{equation}
We have $0<t_i\leq \frac{1}{2}$ for all $i\in [m]$.
Thus, for all $2\leq i<j\leq m$, we have
\[
((t_j+t_{j-1})(t_i+t_{i-1})+(1-t_jt_{j-1})(1-t_it_{i-1}))>0
\quad\implies\quad
w_i^\top w_j>0.
\]
For all $i\in\{2,\dots,m-1\}$, we have
\[
((t_{i+1}+t_i)(t_i+t_{i-1})+(1-t_{i+1}t_i)(1-t_it_{i-1}))=(1+t_i^2)(1+t_{i+1}t_{i-1}).
\]
Together with $t_i-t_{i-1}=1/(2m)$, we obtain
\begin{align*}
w_i^\top w_{i+1}
=\dfrac{4\frac{1}{4m^2}(1+t_{i+1}t_{i-1})}{(1+t_{i+1}^2)(1+t_i^2)(1+t_{i-1}^2)}
\geq \dfrac{1}{m^2(5/4)^3}
=\dfrac{64}{125m^2}.
\end{align*}
Now let $m\geq 2$ and let $\sigma\in \{0,1\}^m$ be not decreasing.
Set $x=\sum_{i=1}^m\sigma_iw_i$.
First suppose that $\sigma_1=1$.
Then, there is an index $k\in \{2,\dots,m-1\}$ with $\sigma_k=0$ and $\sigma_{k+1}=1$.
For every $j\geq 2$, the identity $\|u_j\|_2^2=\|u_{j-1}\|_2^2$ gives
\[
0=\|u_j\|_2^2-\|u_{j-1}\|_2^2=\left\|\sum_{i=1}^jw_i\right\|_2^2-\left\|\sum_{i=1}^{j-1}w_i\right\|_2^2
=\|w_j\|_2^2+2w_1^\top w_j+2\sum_{i=2}^{j-1}w_i^\top w_j.
\]
Inserting this equality and using $\|w_1\|_2^2=1$ and $w_i^\top w_j>0$ for all $2\leq i<j\leq m$, we get
\begin{align*}
\|x\|_2^2
&=1+\sum_{j=2}^m\sigma_j\left(\|w_j\|_2^2+2w_1^\top w_j+2\sum_{i=2}^{j-1}\sigma_iw_i^\top w_j\right)\\
&=1+\sum_{j=2}^m\sigma_j\left(\smash[b]{\|w_j\|_2^2+2w_1^\top w_j+2\sum_{i=2}^{j-1}w_i^\top w_j}-2\sum_{i=2}^{j-1}(1-\sigma_i)w_i^\top w_j\right)\\
&=1-2\sum_{j=2}^m\sigma_j\left(\sum_{i=2}^{j-1}(1-\sigma_i)\smash[b]{w_i^\top w_j}\right)\\
&\leq 1-2w_k^\top w_{k+1}\leq 1-\dfrac{128}{125m^2}\leq 1-\dfrac{16}{25m^2}.
\end{align*}
Now suppose that $\sigma_1=0$.
Then, using that $w_i^\top w_j\geq 0$ holds for all $i,j\in \{2,\dots,m\}$, we have
\[
\|x\|_2^2
= \sum_{i,j\in \{2,\dots,m\}}\sigma_i\sigma_jw_i^\top w_j
\leq \sum_{i,j\in \{2,\dots,m\}}w_i^\top w_j
= \left\|\sum_{i=2}^mw_i\right\|_2^2
= \|u_m-u_1\|_2^2.
\]
Using the second formula in \eqref{eq:gadgetcalc} and $m\geq 2$, we obtain
\[
\|u_m-u_1\|_2^2=\dfrac{4(t_m-t_1)^2}{(1+t_m^2)(1+t_1^2)}
\leq \dfrac{4(1/2)^2}{5/4}
=\dfrac{4}{5}
\leq 1-\dfrac{16}{25m^2},
\]
since $m\geq 2$.
This proves the second claim.

Finally, for $i,j\in [m],\,i\neq j$, we have $|t_i-t_j|\geq 1/(2m)$ and $t_i,t_j\leq 1/2$.
Thus, by \eqref{eq:gadgetcalc},
\[
\|u_i-u_j\|_2^2
=\dfrac{4(t_i-t_j)^2}{(1+t_i^2)(1+t_j^2)}
\geq \dfrac{4\frac{1}{4m^2}}{(5/4)^2}
=\dfrac{16}{25m^2}
=\delta_m.
\qedhere
\]
\end{proof}
Every binary vector $\sigma$ that leads to a point with high $L_2^2$-norm is nonzero and decreasing and corresponds to a unique index $k\in [m]$ with $\sigma_i=1$ for $i\leq k$ and $\sigma_i=0$ for $i>k$.
In the reduction, we use this property of the gadget to encode the choice of one of $m$ objects into binary vectors.

\Cref{lem:gadget} can also be interpreted as a statement about a zonotope with generators $w_1,\dots, w_m$ (see \Cref{fig:gadget}).
The corresponding statement is that precisely the nonzero decreasing binary vectors correspond to vertices of the zonotope that lie on the unit circle, and all other binary vectors yield points in the interior of the unit circle with $L_2^2$-norm at most $1-\delta_m$.

A similar gadget can be constructed for the $L_p$-norms with $p\in (1,\infty)$ by choosing suitable points that lie (approximately) on the $L_p$-unit circle.
The curvature of the $L_p$-unit circle for $p\in (1,\infty)$ again implies that not decreasing binary vectors lead to points that lie some distance away from the $L_p$-unit circle.
This provides a geometric intuition why a gadget cannot be constructed for $p=1$ and $p=\infty$: in these cases, the unit circles consist of polyhedral faces and have no curvature.
As discussed earlier, this property allows for algorithms for \textsc{$L_p^p$-Max on Zonotopes} that run in FPT time with respect to $d$ for $p=1$ and in polynomial time for $p=\infty$.
While it is easy to choose suitable points that lie exactly on the $L_p$-unit circle when allowing irrational points, the same is not always possible for rational points, since already for $p=3$, the only nonnegative rational solutions to $x^3+y^3=1$ are $(0,1)$ and $(1,0)$, as other points would contradict Fermat's last theorem.
One could instead construct an approximate version of the gadget by choosing vectors that lie only approximately on the $L_p$-unit circle, similar to \citep{KKW15}, who first choose potentially irrational points on the $L_p$-unit circle and then show that rounding these points to a sufficiently fine grid preserves the relevant structure of their reduction.
Instead, our transfer from $L_2^2$- to $L_p^p$-maximization for $p\in(1,\infty)$ in \Cref{sec:Lphardness} avoids explicitly adjusting the $L_2$-gadget using another approach.

Before we describe the reduction, we give a high-level overview of the construction.

\textbf{Proof idea.}
For an instance $(G,k)$ of \textsc{Multicolored Clique}, we construct a generator matrix $A_G$ with one generator per node, one generator per edge, and one extra generator.
A binary vector $\sigma\in\{0,1\}^{|V|+|E|+1}$ then ``chooses'' a subset of edges, where an edge is chosen if $\sigma$ has a $1$ in the corresponding entry.
Similarly, using the gadget in \Cref{lem:gadget}, every $\sigma$ with high $\|x_\sigma\|_2^2$ value for $x_\sigma\coloneqq A_G \sigma$ ``chooses'' exactly one node for every color.
Relative to the chosen nodes, choosing an edge whose endpoints are both chosen nodes then increases $\|x_{\sigma}\|_2^2$, whereas choosing an edge where at least one endpoint is not a chosen node decreases $\|x_{\sigma}\|_2^2$.
Thus, once the choice of nodes is fixed, the choice of edges that leads to a maximum $\|x_{\sigma}\|_2^2$ value corresponds to choosing exactly the edges whose endpoints are both chosen nodes.
Hence, if the chosen nodes form a $k$-colored clique, then $\binom{k}{2}$ edges are chosen, which leads to $\|x_{\sigma}\|_2^2$ reaching a threshold value.
If the chosen nodes do not form a $k$-colored clique, at most $\binom{k}{2}-1$ edges are chosen, which leads to $\|x_{\sigma}\|_2^2$ not reaching the threshold.
\begin{theorem}\label{thm:l2norm_hardness}
	\textsc{$L_2^2$-Max on Zonotopes} is W[1]-hard with respect to $d$ and not solvable in $\rho(d) N^{o(d)}$ time (where $N$ is the input bit-length) for any computable function $\rho$ assuming the ETH.
\end{theorem}
\begin{proof}
	We reduce from \textsc{Multicolored Clique}.
	Let $(G=(V=V_1\dot\cup\cdots\dot\cup V_k, E),k)$ be an instance of \textsc{Multicolored Clique} with $k\geq 2$.
    We write $V_i=\{v_{i,1},\dots,v_{i,n_i}\}$, $n= |V|$, and $m=|E|$.
    We assume that $G$ has no edges within a color class and has at least one edge between every pair of color classes ($m\geq \binom{k}{2}$); otherwise, the instance is a no-instance.
	We construct a matrix $A_G$ and~a threshold $L$ such that $\max_{z\in Z(A_G)} \norm{z}_2^2> L$ if $G$ is a yes-instance and $\max_{z\in Z(A_G)} \norm{z}_2^2<L-m^2$ otherwise.
	
	\textbf{Construction.}
    For each color $i\in[k]$, let $c_{i,0}=u^{(i)}_0$ and $c_{i,p}=u^{(i)}_p,\,p\in[n_i]$ be the $n_i+1$ points from \Cref{lem:gadget}.
	For $p,q\in[n_i],\,p\neq q$,
	\[
	\|c_{i,p}-c_{i,q}\|_2^2\geq\frac{16}{25n_i^2}\geq\frac{16}{25n^2}\eqcolon \delta.
	\]
	We write $\R^{2k+1}$ as a Cartesian product of two-dimensional spaces and a single one-dimensional space
	\[
	\R^{2k+1}=S_0\times S_1\times \dots \times S_k,
	\quad S_0\cong \R,\, S_i\cong \R^2 \text{ for }i\in[k].
	\]
	Let $\pi_i$ denote the projection to $S_i$ and let $e_0$ be the standard basis vector corresponding to $S_0$.
	
	Let $\Lambda>0$ and $T>0$ be constants whose roles and values are specified below.
	For every node $v_{i,p}\in V_i$, define a generator $a_{i,p}\in \Q^{2k+1}$ with zero entries everywhere but
	\[
	\pi_i(a_{i,p})=\Lambda(c_{i,p}-c_{i,p-1}).
	\]
	For every edge $e=\{v_{i,p},v_{j,q}\}\in E,\, i<j$, define $b_e\coloneqq b_{i,p,j,q}\in \Q^{2k+1}$ with zero entries everywhere but
	\[
	\pi_i(b_{i,p,j,q})=c_{i,p},
	\quad
	\pi_j(b_{i,p,j,q})=c_{j,q},
	\]
	and define the generator $a_e\in \Q^{2k+1}$ as $a_e=b_e-e_0$.
	Finally, define a generator $a_0=Te_0$.
	The generator matrix $A_G\in\Q^{(2k+1)\times (1+n+m)}$ has as columns $a_0$, all node generators $a_{i,p}$, and all edge generators $a_e$.
	
	\textbf{Value of a binary vector.}
    Let $\sigma\in \{0,1\}^{n+m+1}$ be a binary vector that selects a subset of columns of $A_G$.
	Write $\sigma_0$ for the entry of $a_0$ and $\sigma^{(i)}\in\{0,1\}^{n_i}$ for the entries of the node generators of color $i$.
	Let $F\subseteq E$ be the set of edges with 1-entries in $\sigma$ and let $r\coloneqq |F|$.
	Define
    \[
    s\coloneq \sum_{i=1}^{k}\sum_{p=1}^{n_i}\sigma^{(i)}_p a_{i,p},
    \qquad
    Q_F\coloneq \sum_{e\in F}b_e.
    \]
    Then
	\[
	x\coloneqq A_G \sigma = s + Q_F + (\sigma_0 T-r)e_0,
	\]
	which implies, by the orthogonality of $S_0$ and $S_1\times \dots \times S_k$, that
	\begin{equation}\label{eq:x_norm_expansion}
		\|x\|_2^2
		=\|s + Q_F\|_2^2+(\sigma_0 T-r)^2
        =\|s\|_2^2+\sigma_0 T^2+\Phi(F),
	\end{equation}
    where
    \[\Phi(F)\coloneq\|Q_F\|_2^2+r^2+2s^\top Q_F -2\sigma_0Tr\]
    is the contribution of the chosen edges to $\|x\|_2^2$.
    For $\|x\|_2^2$ to be large, we will show that every individual component $\|s\|_2^2$, $\sigma_0 T^2$, and $\Phi(F)$ must be large.
    
    \textbf{Choice of nodes.}
	First, by \Cref{lem:gadget}, we have
	\[
	\|s\|_2^2
	=\sum_{i=1}^{k}\left\|\pi_i(s)\right\|_2^2
	=\sum_{i=1}^{k}\Lambda^2\left\|\sum_{p=1}^{n_i}\sigma^{(i)}_p (c_{i,p}-c_{i,p-1})\right\|_2^2
	\leq k\Lambda^2,
	\]
    and equality holds exactly when every $\sigma^{(i)}$ is nonzero and decreasing.
    In that case, there is a $p_i\in [n_i]$ with $\sigma^{(i)}_p=1$ exactly for $p\leq p_i$ for every $i\in [k]$ and
	\[
	\pi_i(s)=\sum_{p=1}^{p_i}\Lambda\left(c_{i,p}-c_{i,p-1}\right)=\Lambda c_{i,p_i} \text{ for }i\in [k]
	\quad
	\implies
	\quad
	\|s\|_2^2=k\Lambda^2.
	\]
    Thus, if equality holds, $\sigma$ encodes the choice of exactly one node $v_{i,p_i}$ for every color $i\in[k]$.
    
	\textbf{Choice of edges.} Every edge vector $b_e$ satisfies $\|b_e\|_2^2=2$.
	By the triangle inequality, $\|Q_F\|_2\leq \sum_{e\in F}\|b_e\|_2=r\sqrt{2}$ and
	\[
	\|Q_F\|_2^2\leq 2r^2.
	\]
	Further, for every edge $e=\{v_{i,p},v_{j,q}\}\in F$, we have, using $\|\pi_i(s)\|_2\leq \Lambda$ for every $i\in [k]$,
	\[
	s^\top b_e
	=\pi_i(s)^\top c_{i,p}+\pi_j(s)^\top c_{j,q}
	\leq \|\pi_i(s)\|_2 \|c_{i,p}\|_2+\|\pi_j(s)\|_2 \|c_{j,q}\|_2
	\leq2\Lambda
	\quad
	\implies
	\quad
	2s^\top Q_F\leq 4\Lambda r.
	\]
    Observe that if $\sigma$ encodes the choice of exactly one node $v_{i,p_i}$ for every color $i\in [k]$, then we have $\pi_i(s)=\Lambda c_{i,p_i}$ and $\pi_j(s)=\Lambda c_{j,p_j}$.
    In this case, $s^\top b_e=2\Lambda$ holds with equality exactly when both endpoints of the edge $e$ are chosen nodes, that is, we have $p=p_i$ and $q=p_j$ where $e=\{v_{i,p},v_{j,q}\}$.
    
    \textbf{How a clique maximizes the norm.} We aim to choose $L$, $T$, and $\Lambda$ in such a way that $\sigma_0 =0$ implies $\|x\|_2^2<L-m^2$ (due to $\sigma_0 T^2=0$), which then allows us to restrict to the case where $\sigma_0=1$.
    We have already seen that $\|s\|_2^2$ is maximal if and only if $\sigma$ corresponds to selecting one node $v_{i,p_i}$ for every color $i\in [k]$.
    On the other hand, if $\sigma$ does not correspond to a selection of nodes, then $\|x\|_2^2<L-m^2$ should hold.
    If these criteria are met, then the norm $\|x\|_2^2$ is maximized by making $s^\top b_e$ tight for all the selected edges, that is, by choosing all the edges between selected nodes.
    Only if the selected nodes form a clique, enough edges are chosen and $\|x\|_2^2$ exceeds $L$.
    Choosing an edge $e$ influences the $\Phi(F)$ part of $\|x\|_2^2$ as follows: every edge subtracts $1$ from the $e_0$ coordinate and therefore decreases $(T-r)^2$, which leads to a ``cost'' of selecting the edge; it also increases $\|Q_F\|_2^2$ and $2s^\top Q_F$, which leads to a ``gain'' of selecting the edge.
    The additional contribution of the edge $e$ to $2s^\top Q_F$ is $2s^\top b_e$, which is maximal if the edge lies between selected nodes.
    We will therefore choose the constants $T$ and $\Lambda$ such that the gain is larger than the cost of selecting an edge if and only if the edge lies between selected nodes.
    Now, we show that one can indeed choose the constants $L$, $T$, and $\Lambda$ in such a way to reproduce this behavior.
	
	\textbf{Case 1: $\sigma_0=0$.}
	We use \eqref{eq:x_norm_expansion} together with our derived bounds to obtain
	\[
	\|x\|_2^2
	=\|s\|_2^2  + \|Q_F\|_2^2+r^2+ 2s^\top Q_F
	\leq k\Lambda^2+4\Lambda m+3m^2.
	\]
	Thus, $\|x\|_2^2<L-m^2$ is implied by
	\begin{equation}\label{ieq:case1}
		L-k\Lambda^2-4\Lambda m-4m^2>0.
	\end{equation}
	\textbf{Case 2: $\sigma_0=1$ and some $\sigma^{(i)}$ is zero or not decreasing.}
	In this case, \Cref{lem:gadget} gives $\|s\|_2^2\leq k\Lambda^2-\delta\Lambda^2$.
	Using \eqref{eq:x_norm_expansion} together with our derived bounds, we obtain
	\[
	\|x\|_2^2
	\leq\|s\|_2^2 + T^2 + \|Q_F\|_2^2+r^2+2s^\top Q_F
	\leq k\Lambda^2-\delta\Lambda^2+T^2+3m^2+4\Lambda m.
	\]
	Thus, $\|x\|_2^2<L-m^2$ is implied by
	\begin{equation}\label{ieq:case2}
	L-k\Lambda^2+\delta\Lambda^2-4\Lambda m-4m^2-T^2>0.
	\end{equation}
	\textbf{Case 3: $\sigma_0=1$ and $\sigma^{(i)}$ is nonzero and decreasing for all $i\in [k]$.}
	For every color $i\in[k]$, let $v_{i,p_i}$ be the selected node.
    As discussed earlier, we then have $\pi_i(s)=\Lambda c_{i,p_i}$ for all $i\in [k]$ and $\|s\|_2^2=k\Lambda^2$.
    
	We call an edge $e=\{v_{i,p},v_{j,q}\}$ \emph{selected} if its endpoints are both selected nodes, that is, if $p=p_i$ and $q=p_j$.
	For every edge $e=\{v_{i,p},v_{j,q}\}$, we have
	\begin{equation}\label{eq:matching_vs_nonmatching}
	s^\top b_e
	=\Lambda c_{i,p_i}^\top c_{i,p}+\Lambda c_{j,p_j}^\top c_{j,q}
	\begin{cases}
		=2\Lambda & \text{ if $e$ is selected,}\\
		\leq \Lambda(2-\delta/2)  & \text{ if $e$ is not selected,}
	\end{cases}
	\end{equation}
	where the equality follows from $\|c_{i,p_i}\|_2=\|c_{j,p_j}\|_2=1$, and the inequality follows by \Cref{lem:gadget}, since we have $\min_{p\neq q}\|c_{i,p}-c_{i,q}\|_2^2\geq \delta$ for all $i\in [k]$ and thus, for $p\neq q$,
	\[
	c_{i,p}^\top c_{i,q}
	=\frac{1}{2}\left(\|c_{i,p}\|_2^2+\|c_{i,q}\|_2^2-\|c_{i,p}- c_{i,q}\|_2^2\right)
	=1-\frac{1}{2}\|c_{i,p}- c_{i,q}\|_2^2
	\leq 1-\frac \delta2.
	\]
	For every edge $e$, define $\ell_e\coloneq 2s^\top b_e-2T$.
	Since $r=|F|$, we have $2s^\top Q_F-2Tr=\sum_{e\in F} (2s^\top b_e-2T)=\sum_{e\in F}\ell_e$.
    Therefore, we can rewrite
	\[
	\|x\|_2^2
	=k\Lambda^2+T^2 + \sum_{e\in F}\ell_e + \|Q_F\|_2^2+r^2.
	\]
	For a fixed choice of nodes, $\|x\|_2^2$ has maximum value when $\Phi(F)= \sum_{e\in F}\ell_e + \|Q_F\|_2^2+r^2$ has maximum value for all subsets $F\subseteq E$.
	We will now derive conditions on $T$ and $\Lambda$ that imply, when satisfied, that $\Phi(F)$ is maximized exactly when $F$ is the set of all selected edges.
	
	First, suppose that $F$ contains a non-selected edge $e$ and let $F'\coloneqq F\setminus\{e\}$.
	Then, since $|F'|=r-1$ and $Q_F=Q_{F'}+b_e$, we have
	\[
	\Phi(F)-\Phi(F')
	=\ell_e+(\|Q_{F'}+b_e\|_2^2-\|Q_{F'}\|_2^2)+(r^2-(r-1)^2)
	=\ell_e+2Q_{F'}^\top b_e+\|b_e\|_2^2+2r-1.
	\]
	For any two edges $e,e'\in E$, we have $b_{e'}^\top b_e\leq 2$, since their support overlaps in at most two $S_i$-blocks, and in each common block the inner product of the two unit vectors is at most one.
	Therefore, $Q_{F'}^\top b_e=\sum_{{e'}\in F'}b_{e'}^\top b_e\leq 2(r-1)$.
	Together with $\|b_e\|_2^2=2$ and \eqref{eq:matching_vs_nonmatching}, we obtain 
	\[
	\Phi(F)-\Phi(F')
	\leq \ell_e+4(r-1)+2+2r-1
	= 2s^\top b_e-2T+6r-3
	\leq 4\Lambda -\delta\Lambda -2T+6m-3.
	\]
	Thus, deleting a non-selected edge from $F$ strictly increases $\Phi$ if
	\begin{equation}\label{ieq:case3_nonmatching}
		-4\Lambda +\delta\Lambda +2T-6m+3>0.
	\end{equation}
	In particular, if \eqref{ieq:case3_nonmatching} is satisfied, then no maximizing set $F$ can contain a non-selected edge.
	
	Now suppose that $F$ contains only selected edges, and let $e$ be a selected edge not in $F$.
	Then
	\[
	\Phi(F\cup \{e\})-\Phi(F)
	= \ell_e+(\|Q_{F}+b_e\|_2^2-\|Q_{F}\|_2^2)+((r+1)^2-r^2)
	= \ell_e+2Q_{F}^\top b_e+\|b_e\|_2^2+2r+1.
	\]
	For any two edges $e,{e'}\in E$, we have $b_{e'}^\top b_e\geq -2$ and thus $Q_{F}^\top b_e=\sum_{{e'}\in F}b_{e'}^\top b_e\geq -2r$.
	Together with $\|b_e\|_2^2=2$ and \eqref{eq:matching_vs_nonmatching}, we obtain 
	\[
	\Phi(F\cup \{e\})-\Phi(F)
	\geq \ell_e-4r+2+2r+1
	= 2s^\top b_e-2T-2r+3
	\geq 4\Lambda -2T-2m+3.
	\]
	Thus, adding a missing selected edge to $F$ strictly increases $\Phi$ if
	\begin{equation}\label{ieq:case3_missingmatching}
		4\Lambda -2T-2m+3>0.
	\end{equation}
	Thus, if \eqref{ieq:case3_nonmatching} and \eqref{ieq:case3_missingmatching} hold, then the unique maximizing set $F$ is exactly the set of all selected edges.
	
	\textbf{Choice of constants.}
    We are now ready to define $\Lambda$, $T$, and $L$.
	We set
	\[
	\Lambda\coloneq\frac{100(m+1)^2}{\delta^2},
	\quad
	T\coloneq2\Lambda-\frac{\delta \Lambda}{4},
	\quad
	L\coloneq k\Lambda^2+T^2+\binom{k}{2}\frac{\delta\Lambda}{2}.
	\]
	We have $0<\delta<1$, $1<\Lambda<T$, $4m+4m^2<\Lambda$, and $100(m+1)^2\Lambda<\delta\Lambda^2$.
	Next, we verify that \Cref{ieq:case1,ieq:case2,ieq:case3_missingmatching,ieq:case3_nonmatching} are satisfied.
	\Cref{ieq:case1} is satisfied because
	\begin{align*}
		L-k\Lambda^2-4\Lambda m-4m^2
		>\Lambda^2-4\Lambda m-4m^2
		>\Lambda(4m+4m^2)-4\Lambda m-4m^2
		>4m^2(\Lambda-1)
		>0.
	\end{align*}
	\Cref{ieq:case2} is satisfied since
	\begin{align*}
		L-k\Lambda^2+\delta\Lambda^2-4\Lambda m-4m^2-T^2
		>\delta\Lambda^2-4\Lambda m-4m^2
		>100(m+1)^2\Lambda-4\Lambda m-4m^2
		>0.
	\end{align*}
	For \Cref{ieq:case3_missingmatching,ieq:case3_nonmatching}, we use $4\Lambda -2T=\delta\Lambda/2>50(m+1)^2$.
	\Cref{ieq:case3_nonmatching} is satisfied since
	\[
	-4\Lambda +\delta\Lambda +2T-6m+3
	=\delta\Lambda/2-6m+3
	>50(m+1)^2-6m
	>0.
	\]
	\Cref{ieq:case3_missingmatching} is satisfied since
	\[
	4\Lambda -2T-2m+3
	=\delta\Lambda/2-2m+3
	>50(m+1)^2-2m
	>0.
	\]
	
	\textbf{Correctness.}
    Suppose that $G$ contains a $k$-colored clique $\{v_{1,p_1},\dots, v_{k,p_k}\}$.
	Then, we set $\sigma_0\coloneqq 1$, set $\sigma^{(i)}$ such that the node $v_{i,p_i}$ is chosen for every $i\in[k]$, and set the $\sigma$-entry for every edge of the clique to $1$.
	Let $F\subseteq E$ be the set of these edges.
	We have $r\coloneqq |F|=\binom{k}{2}$.
	For every $e\in F$, \Cref{eq:matching_vs_nonmatching} gives $s^\top b_e=2\Lambda$ and thus $\ell_e=2s^\top b_e-2T=4\Lambda -2T=\frac{\delta\Lambda}{2}$.
	Hence,
	\[
	\|x\|_2^2
	=L + \|Q_F\|_2^2+\binom{k}{2}^2
	>L.
	\]
	Now suppose that $G$ has no $k$-colored clique.
	By Cases 1 and 2, it suffices to consider $\sigma\in \{0,1\}^{n+m+1}$ with $\sigma_0=1$ and where $\sigma$ encodes the choice of one node per color.
	Let $F$ be the set of all selected edges induced by the choice of nodes.
	Since $G$ has no $k$-colored clique, we have $r\coloneqq |F|\leq \binom{k}{2}-1$.
	Using $r\leq m$ and the inequalities $\|Q_F\|_2^2\leq 2m^2$ and $\delta\Lambda/2>50(m+1)^2$, we obtain
	\begin{align*}
		\|x\|_2^2
		&=k\Lambda^2+T^2+r\delta\Lambda/2 + \|Q_F\|_2^2+r^2\\
		&\leq L-\delta\Lambda/2 + 3m^2
		<L-50(m+1)^2+3m^2
		<L-m^2.
	\end{align*}
	By \Cref{lem:l2max_binaryformulation}, $(G,k)$ is a yes-instance if $\max_{z\in Z(A_G)} \norm{z}_2^2> L$ and $\max_{z\in Z(A_G)} \norm{z}_2^2< L-m^2$ otherwise.
	
    Note that the reduction clearly runs in polynomial time and
	the dimension is $d=2k+1\in \mathcal O(k)$, which proves W[1]-hardness.
    Consequently, an algorithm with running time $\rho(d) N^{o(d)}$ would solve \textsc{Multicolored Clique} in time $\rho(2k+1)n^{o(k)}$ (since $N\leq n^{\mathcal{O}(1)}$), contradicting the ETH.
\end{proof}

\end{document}